\documentclass[11pt,a4paper,reqno]{article}
\usepackage[utf8]{inputenc} 
\usepackage[T1]{fontenc}
\usepackage{url}
\usepackage{ifthen}
\usepackage{cite}
\usepackage[cmex10]{amsmath} 

\usepackage{graphicx} 
\usepackage{mathrsfs}
\usepackage{graphicx}
\usepackage{subfigure}
\usepackage{algorithmic, algorithm}
\usepackage[english]{babel}
\usepackage{amssymb}
\usepackage{amsmath}
\usepackage{cite}
\usepackage{graphicx}
\usepackage{amsthm}
\usepackage{thmtools, thm-restate}
\usepackage{hyperref}
\usepackage[letterpaper,margin=0.8in]{geometry}
\usepackage[T1]{fontenc} 
\usepackage[english]{babel} 
\usepackage{amsmath,amsfonts,amsthm,amssymb} 

\usepackage[toc,page]{appendix}
\usepackage{lipsum} 
\usepackage{amsmath, amssymb}
\usepackage[dvipsnames]{xcolor} 
\usepackage{microtype}
\usepackage{color}
\usepackage{cite}
\usepackage{pifont}
\usepackage{enumitem}
\usepackage{graphicx}
\usepackage{lipsum}
\usepackage[dvipsnames]{xcolor}
\usepackage{mathtools}
\usepackage{cleveref}
\usepackage{verbatim}
\usepackage[utf8]{inputenc}
\usepackage[english]{babel}

\usepackage{tikz,pgfplots,lipsum,lmodern}
\usetikzlibrary{decorations.pathreplacing}

\theoremstyle{definition}
\newtheorem{theorem}{Theorem}
\newtheorem{notation}{Notation}
\newtheorem{lemma}{Lemma}
\newtheorem{claim}{Claim}
\newtheorem{definition}{Definition}
\newtheorem{remark}{Remark}
\newtheorem{corollary}{Corollary}
\newtheorem{problem}{Problem}
\newtheorem{construction}{Construction}
\newtheorem{example}{Example}
\newtheorem{proposition}{Proposition}

\newtheorem{observation}{Observation}

\newcommand{\bfe}{{\boldsymbol e}}

\newcommand{\bfg}{{\boldsymbol g}}

\newcommand{\bfx}{{\boldsymbol x}}
\newcommand{\bfy}{{\boldsymbol y}}

\newcommand{\cC}{\mathcal{C}}

\newcommand{\cO}{\mathcal{O}}
\newcommand{\cP}{\mathcal{P}}

\newcommand{\cS}{\mathcal{S}}

\newcommand{{\concatenationSimbole}}{\circ}

\renewcommand{\Bbb}{\mathbb}

\newcommand{\E}{{\Bbb E}}
\newcommand{\F}{{\Bbb F}}

\title{Making it to First: The Random Access Problem in DNA Storage}
\usepackage{authblk}
\author[1]{Avital Boruchovsky}
\author[2]{Ohad Elishco}
\author[3]{Ryan Gabrys}
\author[4]{Anina Gruica}
\author[5]{Itzhak Tamo}
\author[1]{Eitan Yaakobi}

\affil[1]{Technion -- Israel Institute of Technology, Haifa, Israel}
\affil[2]{Ben Gurion University of the Negev, Beer-Sheva, Israel}
\affil[3]{University of California, San Diego, USA}
\affil[4]{Technical University of Denmark, Lyngby, Denmark}
\affil[5]{Tel-Aviv University, Tel-Aviv, Israel}
\date{}                    

\begin{document}

\maketitle
\footnotetext{A. B. and E. Y. are supported by the European Union (DiDAX, 101115134). Views and opinions expressed are those of the author(s) only and do not necessarily reflect those of the European Union or the European Research Council Executive Agency. Neither the European Union nor the granting authority can be held responsible for them. O. E. is supported by the Israel Science Foundation (Grant No. 1789/23). R. G. is supported by NSF Grant CCF2212437. A. G. is supported by the Villum Fonden  Grant VIL”52303”. I. T. is supported by the European Research Council (Grant No. 852953). Email addresses: \texttt{\{avital.bor,yaakobi\}@cs.technion.ac.il, ohadeli@bgu.ac.il, rgabrys@ucsd.edu, anigr@dtu.dk, tamo@tauex.tau.ac.il}}
\begin{abstract}
In this paper, we study the \emph{Random Access Problem} in DNA storage, which addresses the challenge of retrieving a specific information strand from a DNA-based storage system. In this framework, the data is represented by $k$ information strands which represent the data and are encoded into $n$ strands using a linear code. Then, each  sequencing read returns one encoded strand which is chosen uniformly at random. The goal under this paradigm is to design codes that minimize the expected number of reads required to recover an arbitrary information strand. We fully solve the case when $k=2$, showing that the best possible code attains a random access expectation of $1+\frac{2}{\sqrt{2}+1}\approx 0.914\cdot 2$ for $q$ large enough. Moreover, we generalize a construction from~\cite{GMZ24}, specifically to $k=3$, for any value of $k$. Our construction uses \textit{$B_{k-1}$ sequences over $\mathbb{Z}_{q-1}$}, that always exist over large finite fields. We show that for every $k\geq 4$, this generalized construction outperforms all previous constructions in terms of reducing the random access expectation.
\end{abstract}

\section{Introduction}
The exponential growth in data generation has created an unprecedented demand for storage technologies, which current solutions are unable to meet. The gap between data storage demand and the capacity of existing technologies continues to widen at an alarming rate each year~\cite{R22}. Addressing this critical challenge has become a global priority, driving the search for innovative and sustainable alternatives. One particularly promising approach is the use of synthetic DNA as a medium for data storage~\cite{DNA21, M23}.

A typical DNA data storage system consists of three primary components: DNA synthesis, storage containers, and DNA sequencing. In the first step, synthetic DNA strands, known as oligos, are generated to encode the user's information.  These strands are then stored in an unordered manner within a storage container. In the final step, DNA sequencing reads the stored strands and converts them into digital sequences, referred to as {reads}, which are decoded back into the original user information. However, due to limitations in current technologies, the process produces multiple noisy copies of each designed strand, and these copies are retrieved in a completely unordered fashion. 

While several studies have demonstrated the significant potential of DNA as a data storage medium \cite{AVAAY19, BGHCTIPC16, BLCCSS16,OACLY18,YGM17,TYZM15,BOSEY21}, its adoption as a practical alternative to current storage technologies is still limited by challenges related to cost and efficiency. A key factor contributing to these challenges is the coverage depth of DNA storage, defined as the ratio between the number of sequenced reads and the number of designed strands\cite{HMG19}. Reducing the coverage depth is critical for improving the latency of existing DNA storage systems and significantly lowering their costs.

In a recent paper \cite{BSGY24}, Bar-Lev et al. initiated the study of the \emph{DNA Coverage Depth Problem}, which aims to reduce the cost and latency of DNA sequencing by analyzing the expected number of reads required to retrieve the user's information. In this model, $k$ information strands, representing the user's data, are encoded into $n$ strands using a matrix $G\in \mathbb{F}_q^{k\times n}$, often refereed to as a generator matrix of an error-correcting code. The columns of $G$ are then sampled uniformly at random with repetition until the desired information is successfully retrieved (every such sample is referred to as read). The authors considered two variants of this problem: the \emph{non-random access} setting and the \emph{random access} setting. In the non-random access setting, the objective is to recover all $k$ information strands,  i.e., the entire user's information. In contrast, the random access setting focuses on retrieving a single, specific strand out of the $k$ information strands. We also remark that a subsequent work~\cite{AGY24} studied a related model in which groups of strands that collectively represent a single file are decoded, rather than individual strands or the entire dataset.

By drawing a connection to the classical \emph{Coupon Collector’s Problem}~\cite{F57, FGT92, ER61, N60}, the authors of~\cite{BSGY24} showed that when the $k$ information strands are encoded into $n$ strands using an MDS code, then the expected number of reads required to recover all $k$ information strands is $H_n-H_{n-k}$, where $H_i$ is the $i$-th harmonic number. They further proved that no coding scheme can achieve a lower expectation-making MDS codes optimal for minimizing the expected number of reads when the goal is to retrieve the entire data. The non-random access setting was later extended in~\cite{CY24} to accommodate composite DNA letters~\cite{AVAAY19}, and further explored in~\cite{PGYA24, SACZH24} in the context of combinatorial composites of DNA shortmers~\cite{PRYA21}. A related direction was examined in \cite{CTLMKNGWW19}, where the authors analyzed the trade-offs between reading costs, tied directly to coverage depth, and writing costs. Finally,~\cite{BRY25} studied the non–random-access setting for codes over small finite fields.

In this work, we focus on the random access setting, which remains relatively less understood and presents many open questions. Notably, while MDS codes are optimal in the non-random access setting, they yield an expected number of reads equal to~$k$ in the random access case~\cite{BSGY24} (that is, the expected number of reads required to retrieve a specific information strand is $k$). This is as \textit{good} (or \textit{bad}) as not using any code at all and using the identity code, which also gives a random access expectation equal to~$k$. However, as demonstrated in~\cite{BSGY24}, codes with strictly lower expectation do exist. Our goal is to identify and analyze such codes, i.e., those that minimize the expected number of reads required to retrieve any one of the $k$ original strands. Initial steps towards this problem were made in \cite{BSGY24} and \cite{GBRY24}, where constructions, analysis of well-known codes, and bounds on the random access expectation for arbitrary codes were presented. In a later work~\cite{GMZ24}, this setting was explored from a geometric perspective, leading to a construction that outperformed all previously known codes for $k=3$. More recently, the authors of~\cite{BLLLRS25} resolved a conjecture proposed in~\cite{BSGY24} concerning codes of rate $\frac{1}{2}$, and provided a closed-form expression for the random access expectation of the geometric construction  for $k=3$ from~\cite{GMZ24}. Despite these advancements, many essential questions remain unsolved.

In this paper, we address some of the open questions in the random access setting. In particular, we provide a complete solution for the $k=2$ case, finding the smallest possible random access expectation in this case. Moreover, we generalize the construction of~\cite{GMZ24} to arbitrary values of~$k$, and show that the resulting code outperforms all current known constructions for $k\geq4$. 

The rest of the paper is organized as follows. Section~\ref{defenitions} introduces the relevant definitions and formalizes the problem statement. Section~\ref{Sec:k=2} provides a solution to the random access problem for $k=2$.  
In Section~\ref{sec:generalConstruction}, we provide a construction for general~$k$ and then analyze if for all values, showing it outperforms all current constructions. Finally, Section~\ref{sec:conclusion} concludes the paper and outlines several directions for future research.

\section{Preliminaries}\label{defenitions}
Throughout this paper, we adopt the following notations. Let $k$ and $n$ be positive integers with $k\leq n$, and let $q$ be a prime power. Denote by $[k]$ the set $\{1,2,\dots,k\}$ and by $H_n$ the $n$-th harmonic number, i.e., $H_n:=1+1/2+\dots+1/n$. Let $\mathbb{F}_q$ denote the finite field with $q$ elements, and let ${\mathbb{F}_q^k}$ denote the $k$-dimensional vector space over $\F_q$. For $i \in [k]$ denote by ${\bfe}_i\in \F_q^k$ the $i$-th standard vector, and for an integer $\ell \ge 0$ and a set of vectors ${\bf g}_1,{\bf g}_2,\dots,{\bf g}_{\ell} \in \F_q^k$, let $\langle{\bf g}_1,{\bf g}_2,\dots,{\bf g}_{\ell}\rangle$ denote their span. For a set $S$, we denote by $\cP r(S)$ the set of all probability distributions over $S$. 

We study the expected sample size required for uniform random access in DNA storage systems. In such systems, data is stored as a length-$k$ vector of sequences, referred to as strands, each of length~$\ell$ over the alphabet $\Sigma = \{A, C, G, T\}$. This corresponds to representing data as elements in $(\Sigma^\ell)^k$. We can embed $\Sigma^\ell$  into a finite field $\mathbb{F}_q$, which requires $4^\ell$ to divide $q$. However, in this work, we study this problem in a more general setting, considering any prime power $q$ for the size of our underlying finite field. 

The encoding process utilizes a $k$-dimensional linear block code $\cC\subseteq \mathbb{F}_q^n$,  mapping an information vector $\bfx=(x_1,\dots, x_k)\in\mathbb{F}_q^k$ to an encoded vector $\bfy=(y_1,\dots,y_n)\in\mathbb{F}_q^n$. In this storage system, the encoded strands are first synthesized and then sequenced using DNA sequencing technology. This process produces multiple erroneous copies of the strands, referred to as \emph{reads}. In line with prior works~\cite{BSGY24, GBRY24, GMZ24},  we assume that no errors are introduced during synthesis or sequencing. The output of the sequencing process is thus a multiset of unordered reads of encoded strands. Given the high cost and relatively low throughput of current DNA sequencing technologies compared to other archival storage systems, reducing the \emph{coverage depth}—the ratio of sequenced reads to the number of encoded strands—is crucial for improving efficiency.

In this paper, we focus on the random access setting, where the goal is to retrieve a single information strand $\bfx_i$ for $i\in[k]$. Previous works (e.g.,~\cite{BSGY24, BLLLRS25, GBRY24, GMZ24}), have shown that appropriate coding schemes can reduce the expected sample size for recovering an information strand to below $k$.  We illustrate this in Example~\ref{Ex:RAEcanBe<k}, as presented in \cite{GBRY24}. 

\begin{example}\label{Ex:RAEcanBe<k}
Assume we want to store an information vector of length two, $(x_1,x_2)\in \F_q^2$. Without coding, the expected number of samples required to recover a specific information strand is 2, assuming that the samples are read uniformly at random. Now, suppose we encode the data using the following generator matrix:
\begin{align}
    G=\begin{pmatrix}
        1 & 0 &  1 &  0 & 1\\
        0 & 1 &  0 &  1 & 1\\
    \end{pmatrix} \in \F_q^{2 \times 5}. 
\end{align}
The encoded data is stored as $(x_1,x_2)G=(x_1,x_2,x_1,x_2,x_1+x_2)\in \F_q^5$. It can be shown (see Claim~\ref{cl:expcomp} in Section~\ref{Sec:k=2}) that the expected number of samples required to recover a specific information strand for this case is approximately $1.917<2$. Here, “recovering” means reconstructing the original information strand as a linear combination of the sampled symbols. For instance, if the last two encoded symbols are sampled, the strand $x_1$ can be recovered as $x_1 = -x_2 + (x_1 +x_2)$.
\end{example}

Example~\ref{Ex:RAEcanBe<k} demonstrates that once the $k$ information strands are encoded using a generator matrix $G\in \F_q^{k\times n}$, 
each encoded strand corresponds to a column of $G$. Moreover, recovering the $i$-th information strand is equivalent to sampling a collection of columns for which the $i$-th standard basis vector, ${\bfe}_i$, lies in their $\F_q$-span. 

\begin{remark} \label{rem:coll}
Because of the discussion above, we only care about the $\F_q$-span of the sampled columns of the generator matrix $G$. This means in particular, that the order in which the columns show up in the matrix does not matter, and replacing any column by a collinear column will result in the same expected number of samples to recover any information strand.
\end{remark}

We now formally define the main problem addressed in this paper, building on the framework presented in~\cite{GBRY24}.

\begin{problem}[The Random Access Problem]
Let $G\in \F_q^{k\times n}$ be a rank-$k$
matrix. Suppose that the columns of $G$ are drawn uniformly at random with repetition, meaning that each
column has a probability $1/n$ of being drawn and columns can be drawn multiple times. For
$i\in [k]$, let $\tau_i(G)$ denote the random variable that counts the minimum number of columns of $G$ that are drawn, until the standard basis vector ${\bfe}_i \in \F_q^k$ is in their $\F_q$-span. Our goals are:

\begin{enumerate}
    \item Construct full-rank matrices $G \in \F_q^{k \times n}$ for which $$T_{\max}(G) := \max_{i\in[k]} \E[\tau_i(G)] < k.$$
    \item Determine the smallest possible maximum random access expectation among all rank-$k$ matrices: $$T_q(n,k):=\min_{G\in \F_q^{k\times n}} T_{\max}(G),$$
    the asymptotic behavior as $n$ approaches infinity: $$T_q(k) :=\liminf_{n\to \infty} T_q(n, k),$$ and as $q$ approaches infinity:
    $$T(k) :=\liminf_{q\to \infty} T_q(k).$$ 
\end{enumerate}
\end{problem}

Note that the values of $T_q(n,k)$ and $T_q(k)$ are defined only when $q$ is a prime power. If $Q$ is the set of all prime powers, we omit writing $q \in Q$ when writing $q \to +\infty$. In addition, note that if a vector $\bfx$ appears $n_x$ times as a column in $G$, for some integer $n_x \ge 1$, then it has a probability of $\frac{n_x}{n}$ to be selected at each draw. In some cases, it is more convenient to approach the problem from a probabilistic perspective. To this end, we introduce the following related problem.

\begin{problem}
   Let $\cP := \{ \langle v \rangle : v \in \mathbb{F}_q^{\,k}\setminus\{0\}\}$ be the set of one–dimensional subspaces of $\mathbb{F}_q^{\,k}$, that is, the projective points of\/ $PG(k-1,q)$. Select one non-zero representative from each subspace and arrange these vectors as the columns of a matrix  $H\in  \F_q^{k\times |\cP|}$, so that every projective point appears exactly once. For a probability distribution $\mu$ on $\cP$, sample the columns of $H$ with repetition according to $\mu$. For each $i\in [k]$, let $\tau_i(\mu)$ denote the random variable that counts the minimum number of columns of $H$ that are drawn until the standard basis vector ${\bfe}_i \in \F_q^k$ lies in their $\F_q$-span. Our goals are:

\begin{enumerate}
    \item Compute the expectation $\E[\tau_i(\mu)]$ and the maximum expectation: $$T_{\max}(\mu) = \max_{i\in[k]} \E[\tau_i(\mu)].$$
    \item Determine the smallest possible maximum random access expectation among all possible distributions: $$\overline{T}_q(k)=\min_{\mu\in \cP r{[{\cP}]}} T_{\max}(\mu).$$
    
\end{enumerate}
\end{problem}

Problems $1$ and $2$ are closely connected. As observed before stating  Problem 2, any matrix $G\in \F_q^{k\times n}$ induces a probability distribution over $\cP r{[\cP]}$, implying that $\overline{T}_q(k)\leq {T}_q(k)$. Conversely, for any probability distribution in $\cP r{[\cP]}$, there exists a rational probability distribution that is arbitrarily close to it, and for large enough $n$, it can be realized via a matrix $G\in \F_q^{k\times n}$, hence we have that $\overline{T}_q(k)= {T}_q(k)$.  

The study of the values of $T_{\max}(G), \text{~}T_q(n,k),\text{~} T_q(k), $ and $T(k)$ was initiated in~\cite{BSGY24}. It was shown that for several families of codes, such as the identity code, the simple parity code, and MDS codes, it holds that $\E[\tau_i(G)]=k$ for every $i\in[k]$, when $G$ is a systematic generator matrix of these codes. In particular, the result for identity codes established that $T_q(k,k)=k$. However, determining $T_q(n,k)$ for general parameters remains an open and intriguing question. 

Initial progress towards addressing this problem was made in~\cite{BSGY24}, which introduced several code families whose random-access expectation is strictly below $k$. In particular, it was
shown that $T(k=2)\leq 0.914\cdot 2,\text{~} T(k=3)\leq 0.89\cdot 3$ and for arbitrary $k$ which is a multiple of $4$, it holds that for large enough $q$:  $T_q(n=2k,k)\leq 0.95 k$. The authors of~\cite{BSGY24} also conjectured that their rate-$1/2$  codes have a random-access expectation that drops below $0.9456 k$ as $k\to\infty$. This conjecture was later resolved in~\cite{GMZ24}, and the precise asymptotic limit was established in~\cite{BLLLRS25}, proving that $\lim_{k\to\infty}\frac{T_q(n=2k,k)}{k}\leq \frac{8\sqrt{3}\pi-18}{27}$.
Moreover, \cite{BSGY24} proved that for all positive integers $a$ and $k$, it holds that $\frac{T(ak)}{ak}\leq \frac{T(k)}{k}$, showing that the per-strand expectation does not increase when the message length is scaled. In addition, two lower bounds on $T_q(n,k)$ were derived in~\cite{BSGY24}: $T_q(n,k)\geq n-\frac{n(n-k)}{n}(H_n-H_{n-k})$ and $T_q(n,k)\geq \frac{k+1}{2}$.

A significant challenge in the random access problem lies in the difficulty of directly computing $T_{\max}(G)$. To address this, the authors of~\cite{GBRY24} provided a general formula for the expected number of reads required to recover the $i$-th information strand.

\begin{lemma}[see~\textnormal{\cite{GBRY24}, Lemma 1}]\label{Lem:CalExpeUsingAlphas}
For $G\in \F_q^{k\times n}$, let $$\alpha_i^s(G)=|\{S\subseteq [n] : |S|=s, \text{~} {\bfe}_i\in \langle{\bfg}_j:\text{~} j\in S\rangle\},$$ where ${\bfg}_j$ represents the $j$-th column of $G$. Then, for every $i\in[k]$, the expected value of $\tau_i(G)$ is given by:
\begin{align}
\E[\tau_i(G)]=nH_n-\sum_{s=1}^{n-1}\frac{\alpha_i^s(G)}{{n-1 \choose s}}.    
\end{align}
\end{lemma}

Using this result, it was derived that for various families of codes, e.g. Hamming, simplex, systematic MDS, Golay, Reed-Muller, the random access expectation is $k$. Then, in~\cite{GMZ24}, by looking at the random access problem from a geometric point of view, the authors proposed a construction for the case of $k=3$. Using the above formula, they demonstrated that their construction achieves a maximal random-access expectation upper bounded by $0.88\overline{22}\cdot 3$, improving the result of~\cite{BSGY24}. Later, the authors of~\cite{BLLLRS25} computed the exact expectation of this geometric construction and, with optimal parameters, obtained a random access expectation of about $0.881542\cdot 3$. 

Despite these valuable contributions, the fundamental limits of the random access coverage depth problem are still not well understood. Even for the case of $k=2$, the values of $T_q(k)$ and $T(k)$  had not been determined previously. In this work, we seek to deepen and extend the knowledge of the random access problem and provide new insights into the values of $T_{\max}(G), \text{~}T_q(n,k),\text{~} T_q(k) $ and $T(k)$. 
\section{Optimal Random Access Expectation for Two Information Strands}\label{Sec:k=2}
In this section, we determine the exact values of $T_q(2)$ and $T(2)$. To this end, we begin by analyzing the structure of the full-rank matrices that minimize $T_{q}(2)$, and in the process, we adopt the following notation.

\begin{notation}
In this section, unless specified otherwise, $G\in \F_q^{2\times n} $ is a matrix of full rank, and we let $x_1(G)$ and $x_2(G)$ denote the number of columns in $G$ that are equal to $\bfe_1=(1,0)^T$ and $\bfe_2=(0,1)^T$, respectively (recall that these vectors correspond to the information strands). Moreover, over~$\F_q$, there are exactly $q-1$ vectors in $\F_q^2$ that are distinct from the information strands and mutually non-collinear. Due to Remark~\ref{rem:coll}, we only need to consider such non-collinear columns for $G$. Denoting by $\beta \in \F_q \setminus \{0\}$ a primitive element of $\F_q$, the $q-1$ non-collinear vectors can be expressed as $(1,\beta^i)^T$ for $0\leq i\leq q-2$. For $0 \le i \le q-2$ let $a_i(G)$ denote the number of columns in $G$ of the form $(1,\beta^i)^T$. The total number of columns in $G$ is then given by $x(G)\triangleq x_1(G)+x_2(G)+\sum_{i=0}^{q-2}a_i(G)$. When the matrix $G$ is clear from context, we omit specifying $G$ in our notation and simply write $x,x_1,x_2$ and $a_i$ instead of $x(G),x_1(G),x_2(G)$ and $a_i(G)$.
\end{notation}

In the next claim we express the expectations of $\tau_1(G)$ and $\tau_2(G)$ in terms of $x_1(G)$, $x_2(G)$ and the $a_i(G)$'s for $0 \le i \le q-2$. 

\begin{claim} \label{cl:expcomp}
For a matrix $G\in \F_q^{2\times n}$, it holds that:  
  \begin{align}
\E[\tau_1(G)]=1+\frac{x_2}{x-x_2}+\sum_{i=0}^{q-2}\frac{a_i}{x-a_i},\label{Eq:CalOfEx}\\
\E[\tau_2(G)]=1+\frac{x_1}{x-x_1}+\sum_{i=0}^{q-2}\frac{a_i}{x-a_i}.
\end{align}
\end{claim}
\begin{proof}
   We condition on the first draw. If the first draw yields a column of the form $(0,1)^T$, then, upon subsequently drawing a column of a different form, the first column can be recovered. Consequently, given that the first sample was $(0,1)^T$, $\tau_1(G)$ has a geometric distribution with success probability ${(x-x_2)}/{x}$. Analogous reasoning applies if the first draw corresponds to a column of the form $(1,\beta^i)^T$. Therefore, the expected waiting time $\E[\tau_1(G)]$ is given by
\begin{align}
\E[\tau_1(G)]&=\frac{x_1}{x}+\frac{x_2}{x}(1+\frac{x}{x-x_2})+\sum_{i=0}^{q-2}\frac{a_i}{x}(1+\frac{x}{x-a_i})\nonumber\\&=1+\frac{x_2}{x-x_2}+\sum_{i=0}^{q-2}\frac{a_i}{x-a_i}\nonumber.
\end{align}
$E[\tau_2(G)]$ can be computed in an analogous way.
\end{proof}

The following lemma shows that, in order to determine the value of $T_{q}(2)$, one can assume 
$x_1=x_2$.

\begin{lemma}\label{Lem:x_1=x_2}
If $G\in \F_q^{2\times n}$ is a matrix with $x_1(G)>x_2(G)$, then there exists a matrix $G'\in \F_q^{2\times 2n}$ with $T_{\max}(G')<T_{\max}(G)$ and $x_1(G')=x_2(G')$.  
\end{lemma}

\begin{proof}
Denote $x_1=x_1(G),\text{~}x_2=x_2(G), \text{~}x=x(G),\text{~}a_i=a_i(G)$ for $0 \le i \le q-2$, and assume that $x_1>x_2$. From Equation~\eqref{Eq:CalOfEx}, we have
$$T_{\max}(G)=\E[\tau_2(G)]=1+\frac{x_1}{x-x_1}+\sum_{i=0}^{q-2}\frac{a_i}{x-a_i}.$$
Now, consider the concatenated matrix $G'=G\circ \overline{G} \in \F_q^{2\times 2n}$, where $\overline{G}$ is a modified version of $G$, in which each column of the form $(1,0)^T$ is replaced by $(0,1)^T$, and vice versa. 

First, it is clear from the construction of $G'$, that $x_1(G')=x_2(G')$. In addition, we have that:
\begin{align*}
T_{\max}(G')&=1+\frac{x_1+x_2}{2x-x_1-x_2}+\sum_{i=0}^{q-2}\frac{2a_i}{2x-2a_i}\\&= 1+\frac{x_1+x_2}{2x-x_1-x_2}+\sum_{i=0}^{q-2}\frac{a_i}{x-a_i}.
\end{align*}
Since $x_1>x_2$, it follows that $$\frac{x_1+x_2}{2x-x_1-x_2}<\frac{x_1}{x-x_1},$$and thus $T_{\max}(G')<T_{\max}(G)$. 
\end{proof}

The next Lemma is the analogs of Lemma~\ref{Lem:x_1=x_2} for the $a_i's$.

\begin{lemma}\label{lem:equal_a_i}
    If $G\in \F_q^{2\times n}$ is a matrix with $a_i<a_j$ for some $i,j \in \{0,1,\dots,q-2\}$, then there exists a matrix $G'\in \F_q^{2\times 2n}$ with $T_{\max}(G')<T_{\max}(G)$.
\end{lemma}

\begin{proof}
    Similarly to the proof of Lemma~\ref{Lem:x_1=x_2}, define the matrix $G'=G\circ \overline{G} \in \F_q^{2\times 2n}$, where $\overline{G}$ is a modified version of $G$, in which each column of the form $(1,\alpha^i)^T$ is replaced by $(1,\alpha^j)^T$, and vice versa. Then we have \begin{align*}
    T_{\max}(G')-T_{\max}(G)=\frac{2a_i+2a_j}{2x-a_i-a_j}-\frac{a_i}{x-a_i}-\frac{a_j}{x-a_j}=x(-a_i^2-a_j^2+2a_ia_j).
    \end{align*}
    Thus, since $a_i\neq a_j$, we have that $T_{\max}(G')-T_{\max}(G)<0$.
\end{proof}

Note that,  both Lemma~\ref{Lem:x_1=x_2} and Lemma~\ref{lem:equal_a_i} reduce not only the maximum random-access expectation but also the average. Specifically, in both lemmas it holds that $\frac{\E[\tau_1(G')+\E[\tau_2(G')]]}{2}\leq\frac{\E[\tau_1(G)+\E[\tau_2(G)]]}{2}.$

Now, by Lemmas~\ref{Lem:x_1=x_2} and~\ref{lem:equal_a_i}, to determine the value of $T_q(2)$,  it is sufficient to consider matrices $G$, for which $x_1=x_2$ and $a_i=a_j$ for all $i,j \in \{0,\dots,q-2\}$. Writing $x_1=x_2$, $a:=a_i$, and  $x=2x_1+(q-1)a$, Equation~(\ref{Eq:CalOfEx}) simplifies to
\begin{align}
&T_{\max}(G) = \E[\tau_1(G)]=1+\frac{x_1}{x_1+(q-1)a}+\frac{(q-1)a}{2x_1+(q-2)a} \label{Eq:ExForX_1=X_2andA_i=a}.
\end{align}

The next theorem leverages this formula to determine the exact value of $T_q(2)$.
\begin{theorem} \label{thm:tq2}
We have that $T_q(2)=1+\frac{2q^2-q(\sqrt{2}+1)-2+\sqrt{2}}{q^2(1+\sqrt{2})-q(2+\sqrt{2})}$.
\end{theorem}

\begin{proof}
    We begin by finding the optimal value of $a$ as a function of $x_1$. Taking the derivative of Equation~(\ref{Eq:ExForX_1=X_2andA_i=a}) with respect to $a$ and setting the derivative to zero yields the condition 
\begin{align}
-\frac{x_1(q-1)}{(x_1+(q-1)a)^2}+\frac{(q-1)(2x_1+(q-2)a)-(q-1)(q-2)a}{(2x_1+(q-2)a)^2}=0. \label{Eq:Deriative}  
\end{align}
Solving Equation (\ref{Eq:Deriative}),we find that the optimal value of $a$ is given by \begin{align}
a^*=\frac{\sqrt{2}q-2}{q^2-2}x_1.\label{Eq:OptimalVal_A}    
\end{align}
This choice indeed minimizes the expression in Equation (\ref{Eq:ExForX_1=X_2andA_i=a}), since the the derivative is negative immediately before $a^*$ and positive immediately after it.
Although $a^*$ is not generally an integer, its continuous dependence on $x_1$, allows us to select $x_1$ so that $a^*$ can be made arbitrarily close to an integer.

Substituting $a^*$ from Equation (\ref{Eq:OptimalVal_A}) into Equation (\ref{Eq:ExForX_1=X_2andA_i=a}) yields
\begin{align*}
    T_q(2)&=1+\frac{q^2-2}{q^2-2+(q-1)(\sqrt{2}q-2)}+\frac{q-1}{\frac{2(q^2-2)}{\sqrt{2}q-2}+q-2}\\&=1+\frac{q^2-2}{q^2(1+\sqrt{2})-q(2+\sqrt{2})}+\frac{\sqrt{2}q^2-q(2+\sqrt{2})+2}{q^2(2+\sqrt{2})-q(2+2\sqrt{2})}\\&=1+\frac{q^2-2}{q^2(1+\sqrt{2})-q(2+\sqrt{2})}+\frac{q^2-q(\sqrt{2}+1)+\sqrt{2}}{q^2(\sqrt{2}+1)-q(\sqrt{2}+2)}\\&=1+\frac{2q^2-q(\sqrt{2}+1)-2+\sqrt{2}}{q^2(1+\sqrt{2})-q(2+\sqrt{2})}. \qedhere
    \end{align*} 
\end{proof}

\begin{figure}[ht!]
\centering
\begin{tikzpicture}[scale=1]
\begin{axis}[legend style={at={(1,1)}, legend style={cells={align=left}}, anchor = north east, /tikz/column 2/.style={
                column sep=5pt}},
		legend cell align={left},
		width=15cm,height=7cm,
    xlabel={q},
    ylabel={$T_q(2)/2$},
    xmin=2, xmax=157,
    ymin=0.913, ymax=0.958,
    xtick={2,20,40,60,80,100,120,140,157},
    ytick={0.914,0.920,0.930,0.940,0.950,0.958},
    ymajorgrids=true,
    grid style=dashed,
    every axis plot/.append style={},  yticklabel style={/pgf/number format/fixed}
]
\addplot+[color=orange,mark=o,mark size=1pt,smooth]
coordinates {
(2,0.957106781186547524400844362104)
(3,0.942809041582063365867792482806)
(4,0.935660171779821286601266543156)
(5,0.931370849898476039041350979368)
(7,0.926468767748367184687161763609)
(8,0.924936867076458167701477633683)
(9,0.923745388776084487823723310408)
(11,0.922012329430086408001535203827)
(13,0.920812519113626198893866514655)
(16,0.919575214724776608251583178947)
(17,0.919259823409971810636883505139)
(19,0.918728638037668993601599843989)
(23,0.917943407487308307549441388374)
(25,0.917645019878171246849621175241)
(27,0.917390837840758195142366919609)
(29,0.917171715394712460911975319927)
(31,0.916980866812672627872601991170)
(32,0.916894388548935828526635951578)
(37,0.916532114741389777212453893825)
(41,0.916305914510336632977257291911)
(43,0.916208595806278884876068056205)
(47,0.916038805726858983933567687524)
(49,0.915964305998133925356756301267)
(53,0.915832174403791368635619125640)
(59,0.915667569790500217466066881426)
(61,0.915619897416159064395103663157)
(64,0.915553975461015438664162337893)
(67,0.915493956964541391356887399968)
(71,0.915421822057981034029833953446)
(73,0.915388719052915664571528604700)
(79,0.915299466646853339323186335295)
(81,0.915272654195649430915248122676)
(83,0.915247133910768602430584040785)
(89,0.915177454930700722410658513826)
(97,0.915097958637289945205795026022)
(101,0.915062933042668365150186855653)
(103,0.915046440408307718230798542421)
(107,0.915015304780823132457747708095)
(109,0.915000593910956562115434699217)
(113,0.914972734387492437750346346119)
(121,0.914922541196457899638038404175)
(125,0.914899853874110288411275214415)
(127,0.914889046133936820070966765751)
(128,0.914883768917055243732925531052)
(131,0.914868420675590506444423924788)
(137,0.914839740749933771073209244470)
(139,0.914830730989115947731172977992)
(149,0.914789310276631323641945846866)
(151,0.914781684476584485564591447890)
(157,0.914759972803839666325244846985)
};
\end{axis}
\end{tikzpicture}
\caption{\label{fig:tq2} Normalized (by $k=2$) random access expectation $T_q(2)$ from Theorem~\ref{thm:tq2} and various prime powers $q$.}
\end{figure}

We include a plot (see Figure~\ref{fig:tq2}) to show what the random access expectation from Theorem~\ref{thm:tq2} is for various values of $q$.
Letting $q\to\infty$ gives the following corollary, which confirms the optimality of the construction  in \cite[Theorem 12]{BSGY24}.
\begin{corollary}
It holds that $T(2)=\liminf_{q\to \infty}T_q(2)=1+\frac{2}{\sqrt{2}+1}\approx 0.914\cdot 2$.
\end{corollary}

\section{Record-Low Random-Access Expectation}\label{sec:generalConstruction}
In this section, we will construct a rank-$k$ matrix, which we will denote by $G_k$, with a random access expectation that improves upon all previously known results for all $k\geq 4$. This construction is a generalization of the construction presented in~\cite{GMZ24}, which was specifically given for $k=3$. In the construction, all columns of the matrix $G_k$ are of weight $1$ or $2$. 

Accordingly, the columns of $G_k$ can be naturally interpreted as representing the vertices and edges of $K_k$, the complete graph on $k$ vertices in the following way: If we denote by $V_i$, $i \in [k]$ the vertices of $K_k$, then, for $i \in [k]$, the set of columns of weight $1$ with support $i$ correspond to the vertex $V_i$, and for $i,j \in [k]$ with $i \ne j$, the set of columns of weight $2$ with support $i,j$ correspond to the edge between $V_i$ and $V_j$. This perspective will prove useful in Section~\ref{Sec: AsymptoticExpectationCon}, where we analyze the random access expectation of the proposed construction. 

In Section~\ref{Sec:ConsRecCom}, we present the formal construction of $G_k$, but first we introduce an additional notation and definition. For $i, j\in[k]$, $i \ne j$, let $E_{i,j}$ denote the set of columns of weight $2$ with support $\{i,j\}$ in $G_k$. Our objective is to ensure that $G_k$ satisfies a property we call \emph{recovery completeness}, which we formally define below.

\begin{definition}\label{def:rc}
    A rank-$k$ matrix $G \in \F_q^{k \times n}$ is called \textit{\textbf{recovery complete}}  if the following conditions hold:
\begin{enumerate}
\item For any $i, j \in [k]$, given two distinct columns  ${\bfg},{\bfg}'\in E_{i,j}$, it is possible to recover the basis vectors ${\bfe_{i}}$ and $\bfe_{j}$, i.e., $\{{\bfe}_i,{\bfe}_j\}\subseteq \langle {\bfg},{\bfg}'\rangle$.
\item Suppose we collect one column each from $E_{j_1, j_2}, E_{j_2, j_3}, \ldots, E_{j_m, j_1}$ so all the $j_i$'s are different and the edges $\{j_1, j_2\}, \ldots, \{j_m, j_1 \}$ form a cycle in $K_k$ of length $m \leq k$ . Then, it is possible to recover the basis vectors corresponding to the indices ${j_1}, \ldots, {j_m}$.
\end{enumerate}
\end{definition}

As we will show later in Secion~\ref{Sec: AsymptoticExpectationCon}, recovery complete matrices have a \emph{good} random access expectation. But first, we construct such matrices in the next section.

\subsection{Construction of Recovery Complete Matrices}\label{Sec:ConsRecCom}

We begin by revisiting the case where $k=3$. While a construction for this case, based on a geometric viewpoint, was presented in~\cite{GMZ24}, the approach taken here is designed so that it can be applied to a more general framework. Our matrices for the case where $k=3$ will have the following form:

\begin{align}\label{eq:constr3}
G_3(x) =
\left[
\begin{array}{ccc|cccc|cccc|cccc}
1 & 0 & 0 & 1 & 1 & \cdots & 1 & 1 & 1 & \cdots & 1 & 0 & 0 & \cdots & 0\\
0 & 1 & 0 & \beta^{i_1} & \beta^{i_2} & \cdots & \beta^{i_x} & 0 & 0 & \cdots & 0 & 1 & 1 & \cdots & 1\\
0 & 0 & 1 & 0 & 0 & \cdots & 0 & \beta^{i_{x+1}} & \beta^{i_{x+2}} & \cdots & \beta^{i_{2x}} & \beta^{i_{2x+1}} & \beta^{i_{2x+2}} & \cdots & \beta^{i_{3x}}
\end{array}
\right]
\end{align}
\vspace{-0.7cm}
\[
\hspace{3.3cm}
\underbrace{\hspace{3.3cm}}_{E_{1,2}} \hspace{0.2cm}
\underbrace{\hspace{4cm}}_{E_{1,3}} \hspace{0.2cm}
\underbrace{\hspace{4.4cm}}_{E_{2,3}}
\]
where $G_3(x) \in \F_q^{3 \times (3+3x)}$, $\beta$ is a primitive element of $\F_q$ and all the $i_j$'s are different.

Our objective is to show that for large enough $q$, we can find $\{i_1,\dots,i_{3x}\} \subseteq \{0,\dots,q-2\}$ such that
$G_3(x)$ is recovery complete, with each $E_{i,j}$ containing $x$ columns. To simplify, we assume $\mathbb{F}_q$ is a field of characteristic 2, allowing us to disregard signs.

\begin{claim}
Let $\{i_1,\dots,i_{3x}\} \subseteq \{0,\dots,q-2\}$ be a set of distinct integers such that for all distinct $s,r,\ell \in \{i_1,\dots,i_{3x}\}$, we have $$s \not\equiv r+\ell \bmod (q-1).$$ 
Then, the matrix $G_3(x)$, as described in~\eqref{eq:constr3}, is recovery complete.
\end{claim}
\begin{proof}
To verify the first recovery completeness property, we check that any two columns in $E_{i,j}$ are linearly independent. For any two columns retrieved from $E_{i,j}$, the corresponding submatrix formed by their non-zero rows has the form: 
\begin{align*}
\begin{bmatrix}
1 & 1 \\
\beta^{r} & \beta^{\ell}
\end{bmatrix}
\end{align*}
for some $r,\ell \in \{i_1,\dots,i_{3x}\}$. 
Since the determinant of this matrix is non-zero whenever $r \not\equiv \ell \bmod (q-1)$,  and we assumed that all the powers of $\beta$ showing up in $G_3(x)$ are different and in $\{0,\dots,q-2\}$, the first property of recovery completeness trivially follows.

For the second property of recovery completeness, consider a matrix formed by selecting one column each from $E_{1,2}, E_{1,3}$, and $E_{2,3}$. Note that such a matrix has the following form (up to permutation on the indices):
\begin{align*}
\begin{bmatrix}
1 & 0 & 1 \\
\beta^{r} & 1 & 0 \\
0 & \beta^{ \ell} & \beta^{s}
\end{bmatrix}.
\end{align*}
The determinant of this matrix is $ \beta^{s} + \beta^{r + \ell}$, therefore, to ensure this determinant is non-zero, it suffices to ensure that  $\beta^{s}\neq \beta^{r + \ell}$ or equivalently (as $\mathbb{F}_q$ is a field of characteristic 2) that
\begin{align}\label{eq:sset}
s\not\equiv r+\ell \bmod (q-1).
\end{align}
By the assumptions in the claim, the condition in~\eqref{eq:sset} holds.
\end{proof}

The requirement in (\ref{eq:sset}) is identical to the construction of \textit{sum-free sets}~\cite{DY69,GZ05}; A sum-free set $\cS$ is a set for which there are no solutions to the equation  $r+\ell=s$ with $r,\ell,s\in \cS$. For the setup where the sum-free set is defined over an abelian group of order $n$ (in our case $n=q-1$), it is known that if  $n\not\equiv 0 \bmod 2$ (as is the case here, since $q$ is assumed to be a power of $2$), then there exists a sum-free set of size at least $n(1/3-\frac{1}{3n})$.

Let $\cS \subseteq \{0,\dots,q-2\}$ be a sum-free set of size $3x$. Setting the powers of $\beta$ in $E_{1,2}$, $E_{1,3}$, and $E_{2,3}$ to be equal to the elements of $\cS$, yields the desired result. Although the requirement $x\leq \frac{q-1}{2}$ in~\cite{GMZ24}, which is optimal, is superior to ours,  our construction still ensures $q=\cO(x)$, and can be naturally extended to larger values of $k>3$, as demonstrated in Construction~\ref{ConsGeneralK}.


\begin{construction}\label{ConsGeneralK}
Let $\cS$ be a set of size at least $x \binom{k}{2}$  with the following property: for any $k' \leq k$ distinct elements  $i_1, i_2, \ldots, i_{k'}\in \cS$ (i.e., $i_j\neq i_{j'}$ for $j\neq j'$), the following holds: 
\begin{align}\label{eq:prop}
\ell_1i_1  + \ell_2i_2  + \cdots + \ell_{k'}i_{k'} \not\equiv 0 \bmod (q-1),
\end{align}
where $\ell_i \in \{-1,1\}$ for all $i \in [k']$, with at least one $\ell_i$ positive and another negative. 

The generator matrix for our code is then given by: 
\begin{align}\label{eq:gG}
G_k(x)= \left[\begin{array}{cccccccccc} 
 I_k & E_{1,2} & E_{1,3} & \cdots & E_{1,k} & E_{2,3} & \cdots & E_{2,k} & \cdots & E_{k-1,k}
\end{array} \right]
\end{align}
where there are $\binom{k}{2}$ sub-matrices $E_{i,j} \in \F_q^{k \times x}$, each corresponding to some $i,j \in [k]$ with $i \ne j$. Each submatrix $E_{i,j}$ has the following properties:
\begin{enumerate}
    \item $|E_{i,j}|=x,$
    \item the support of each column in $E_{i,j}$ is $\{i,j\}$,
    \item the leading coefficient of each column in $E_{i,j}$ is $1$,
    \item the second coefficient in each column of $E_{i,j}$ is a power of a fixed primitive element $\beta$, where the powers of $\beta$ in $E_{1,2}$ are the first $x$ elements of $\cS$, the powers of $\beta$ in $E_{1,3}$ are the next $x$ elements of~$\cS$, and so on. 
\end{enumerate}
\end{construction}

Construction~\ref{ConsGeneralK} relies on the existence of so-called \textit{$B_{k-1}$ sequences} over $\mathbb{Z}_{q-1}$, which are sequences satisfying the properties given in the following theorem. Note that we restate the theorem in the language of this paper.

\begin{theorem}[\textnormal{see~\cite[Theorem 2]{BC1960}}] \label{thm:bl}
Let $q=2^{tk}$ for some $t \ge 1$. Then, there exist $2^t+1$ integers $j_0=0,j_1=1,j_2,\dots,j_{2^t+1}\le q-2$ such that the sums
\begin{align*}
    j_{i_1}+\dots +j_{i_{k-1}} \quad (0 \le i_1 \le \dots \le i_{k-1} \le 2^t)
\end{align*}
are all different mod $q-1$.
\end{theorem}


\begin{lemma}
Let $q =  2^{tk}$, where $t=\lceil\log_2(x\binom{k}{2}-1)\rceil$. Then, there exists a set $\cS$ as described in Construction~\ref{ConsGeneralK}, and the generator matrix $G_k(x)$ in~(\ref{eq:gG}) is recovery complete.
\end{lemma}

\begin{proof}
We begin by constructing a set $\cS$ as described in Construction~\ref{ConsGeneralK}. Let $t$ be the smallest integer such that $x\binom{k}{2}\leq 2^t+1$. From Theorem~\ref{thm:bl} there exists a $B_{k-1}$ sequence over $\mathbb{Z}_{q-1}$ of size $2^t+1$ whenever $q =2^{tk}$. We can then choose a subset $\cS$ of this sequence of size $x\binom{k}{2}$ with $0 \in \cS$. 

Assume toward contradiction, that there exist $k'$ different elements $\{i_1,\dots, i_{k'}\}\in \cS$ for which Equation~\eqref{eq:prop} holds with equality. We can then split the sum among the terms corresponding to positive $\ell_i$'s, and negative $\ell_i$'s, respectively, resulting in two different sums of (at most) $k-1$ terms which are equal mod $q-1$.  As $0\in \cS$, we can pad both sums with zeroes to reach exactly $k-1$ terms. This contradicts the assumption that~$\cS$ is a $B_{k-1}$ sequence over $\mathbb{Z}_{q-1}$. 

We now show that $G_k(x)$ is recovery complete. The fact that two distinct columns from $E_{i,j}$ are linearly independent follows from the same reasoning as in the case of $k=3$. 

To prove the second condition of recovery completeness, consider the scenario where we draw one column each from $E_{i_1,i_2}, E_{i_2, i_3}, $ $\dots, E_{i_m, i_1},$ and where $\{i_1, i_2\}, \{i_2, i_3\}$, $\ldots$, $\{i_m, i_1\}$ is a cycle of length $m \leq k$. To satisfy the second property, we have to show that the matrix made of the corresponding columns of $G_k(x)$ is full rank, i.e., has rank $m$. In order to do so, note that we can delete all 0-rows, and permute rows in a suitable way, so that we end up with a matrix of the following form:

\begin{align}\label{eq:gmat}
\begin{bmatrix}
1 & 0 & 0 & \cdots & 0 & 1 \\
\beta^{i_{1} } & a_{i_{2}} & 0 & \cdots & 0 & 0  \\
0 & b_{i_{2}} & a_{i_{3}} & \cdots & 0 & 0 \\
0 &0 & b_{i_{3}} & \cdots & 0 & 0 \\
  & & & \ddots & &\\
0 & 0 & 0 & \cdots & a_{i_{m-1}} & 0\\
0 & 0 & 0 & \cdots & b_{i_{m-1}} & \beta^{i_{m}}
\end{bmatrix},
\end{align}
where in every column $1<j<m$, $\{a_{i_{j}}, b_{i_{j}}\} = \{1,\beta^{i_{j}}\}$, i.e., one of the non-zero entries is $1$ and the other is a power of $\beta$ of the form $\beta^{i_j}$, where $i_j \in \cS$. Since our underlying field has characteristic $2$, the determinant of this matrix is given by:
\begin{align*}
\beta^{i_m}\prod_{j=2}^{m-1}a_{i_j}+\beta^{i_1}\prod_{j=2}^{m-1}b_{i_j}.
\end{align*}
Let $T\subseteq \{2,3,\dots,m-1\}$ denote the subset such that $a_{i_j}=1$ for every $j\in T$, and $a_{i_j}=\beta^{i_j}$ for  every $j\notin T$. Then the determinant becomes:
\begin{align*}
\beta^{i_m+\sum_{j\notin T}i_j}+\beta^{i_1+\sum_{j\in T}i_j}.
\end{align*}
To ensure that this expression is nonzero, it suffices to require:
\begin{align*}
i_m+\sum_{j\notin T}i_j\not\equiv  i_1+\sum_{j\in T}i_j \bmod (q-1),
\end{align*}
where all $i_j$ are distinct.
This condition follows directly from equation (\ref{eq:prop}), which is satisfied by the set $S$. Therefore, the matrix in (\ref{eq:gmat}) has full rank, completing the proof that the generator matrix $G_k(x)$ from (\ref{eq:gG}) is recovery complete.
\end{proof}

In the construction described thus far, each information strand (i.e., column of weight 1) appears exactly once in $G_k(x)$. However, by repeating these columns multiple times, we can reduce the random access expectation. Importantly, such repetition does not affect the recovery completeness of the matrix.
Henceforth, let $G_k(x,y)$ denote the matrix obtained from $G_k(x)$ by appending $y-1$ copies of the identity matrix $I_k$ to $G_k(x)$, resulting in a matrix with $y$ copies of each column of weight 1.

In the next section, we will evaluate the asymptotic random access expectation of $G_k(x,y)$ as $x,y\to {}\infty$, which will enable us to derive an upper bound on $T(k)$.

\subsection{The Random Access Expectation of the Construction}\label{Sec: AsymptoticExpectationCon}
Before analyzing the random access expectation of $G_k(x,y)$, 
we will find it useful to interpret the recovery process as one that is performed over the complete graph $K_k$ (see the beginning of Section~\ref{sec:generalConstruction}) rather than over the columns of the matrix $G_k(x,y)$. More precisely, for each collected column $\bf g$ from $G_k(x,y)$, one of the following corresponding events occurs in $K_k$:
\begin{enumerate}
\item If the column $\bf g$ has weight $1$ with support $j$, then we have collected the vertex labeled $V_j$ in the graph $K_k$.
\item If the column $\bf g$ has weight $2$ with support $\{j_1,j_2\}$, then we have collected the edge $E_{j_1,j_2}$ in $K_k$.
\end{enumerate}
For shorthand, we denote by $\mathscr{T}$ the multiset of collected vertices and edges from $K_k$.

\begin{example}
For the case of $k=3$, the top-left graph in Figure~\ref{fig:k3} illustrates the scenario where a column with a single 1 in the first coordinate is collected, corresponding to $\mathscr{T} = \{ V_1 \}$. The bottom-left graph represents the event where three columns are collected: one with support $\{1,2\}$, another with support $\{2,3\}$, and the third with support $\{1,3\}$, resulting in $\mathscr{T} = \{ E_{1,2}, E_{2,3}, E_{1,3} \}$.  
\end{example}

\begin{figure}[h]
\centering
    \includegraphics[scale=.4]{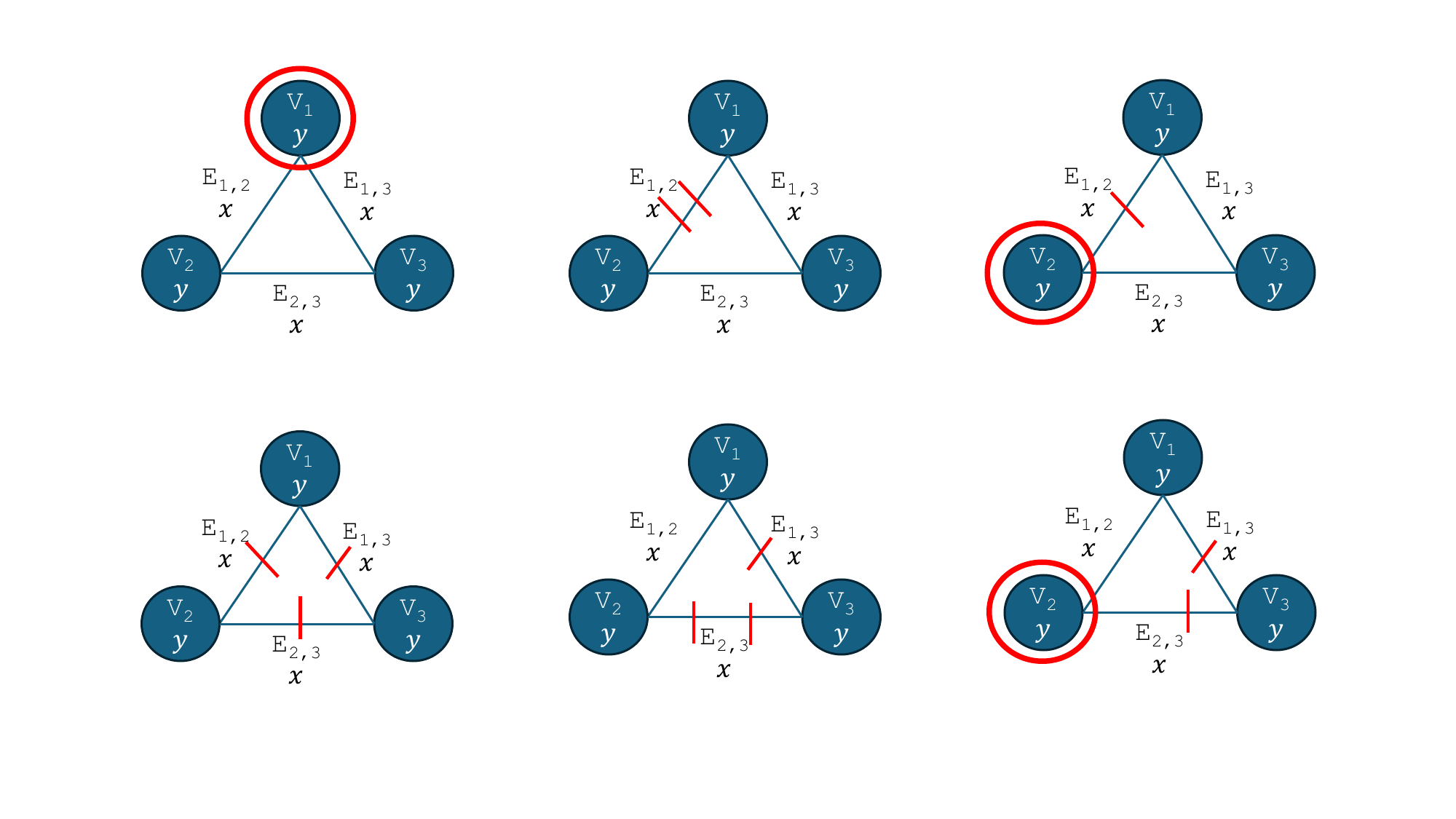}
    \caption{Recovery sets over $K_3$ for information strand $1$.}\label{fig:k3}
    \end{figure}
    
 For convenience, we will refer to the scenario where a vertex is collected (i.e., where $V_j \in \mathscr{T}$) as a cycle of length $1$, and when an edge is collected twice (when the two columns representing this edge are not collinear) as a cycle of length $2$. 
Using this interpretation, we can now state precise conditions under which any information strand in $G_k(x,y)$, or equivalently, any vertex of $K_k$, can be recovered. 
\begin{claim}\label{Claim:EqToGraphs} The information strand corresponding to the vertex $V_i \in K_k$ can be recovered if and only if $V_i$ belongs to a connected component in $\mathscr{T}$ that contains a cycle. 
\end{claim}
\begin{proof}
Assume we want to recover the vertex $V_i$. By the recovery-completeness property, every vertex on the cycle is recoverable. Let $P$ be a path in $\mathscr{T}$ from $V_i$ to the cycle, and let $u$ be the last vertex of 
$P$ that lies on the cycle. Using $u$ and the edge of $P$ incident to $u$, we can recover the preceding vertex on $P$. Repeating this backtracking step recovers each earlier vertex along $P$ until we recover $V_i$.

Conversely, if $V_i$ does not belong to a connected component of $\mathscr{T}$ that contains a cycle, then it cannot be recovered.
\end{proof}
Note that in the previous claim, we say that $V_i$ and $V_j$ belong to the same connected component if there exists a path from $V_i$ to $V_j$ that traverses only edges contained in the multiset $\mathscr{T}$.
It is straightforward to verify that for each of the $6$ scenarios depicted in Figure~\ref{fig:k3}, $V_1$ belongs to a connected component with at least one cycle and similarly for Figure~\ref{fig:k4}. Note that the scenarios in Figure~\ref{fig:k4} are not all of the possibilities for recovering $V_1$ in $K_4$.

\begin{figure}[h]
\centering
    \includegraphics[scale=.4]{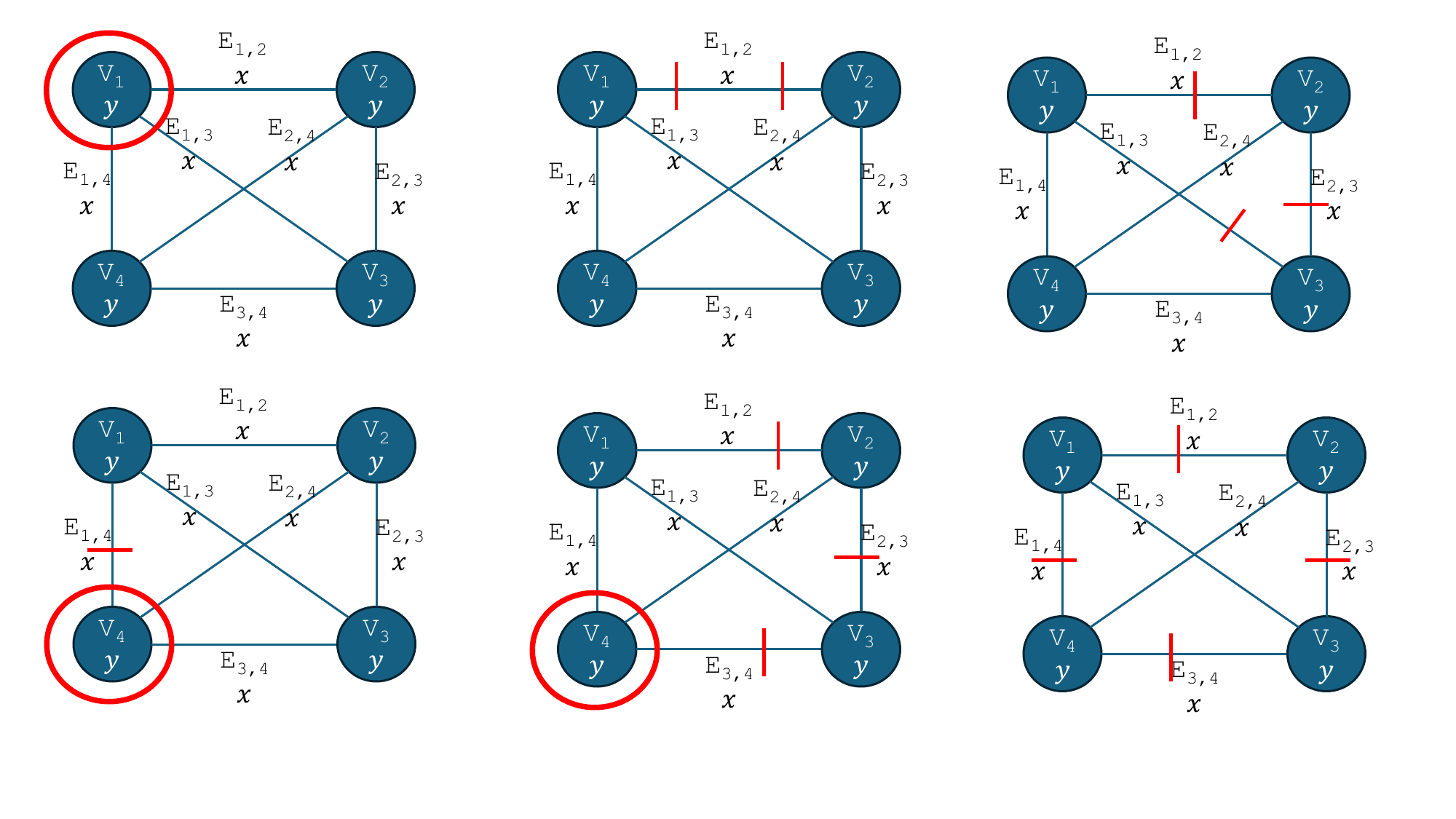}
    \caption{Recovery sets over $K_4$ for information strand $1$.}\label{fig:k4}
    \end{figure}

For general values of $k$, $x$, and $y$, explicitly computing  $T_{\max}(G_k({x,y}))$ is quite challenging. However, in~\cite{GMZ24}, the authors managed to derive an upper bound on $\lim_{x \to \infty}T_{\max}(G_3({x,x}))$, which at the time gave the best upper bound on $T(3)$. Subsequently, \cite{BLLLRS25} determined the exact limit of $\lim_{x \to \infty}T_{\max}(G_3({x,\alpha x}))$, therefore improving the result of~\cite{GMZ24}. As \cite{BLLLRS25} was published after the first draft of this paper appeared online, we include in the next subsection our own improved upper bound on $\lim_{x \to \infty}T_{\max}(G_3({x,\alpha x}))$ (which sharpened~\cite{GMZ24}) and, for completeness, we also summarize the results of~\cite{BLLLRS25}.

\subsubsection{Three Information Strands}
In~\cite{GMZ24}, the authors presented a construction of $G_3(x,y)$ based on a geometric approach. They analyzed the asymptotic behavior of the random access expectation for the case in which the multiplicity of the information strands equals that of the strands per edge, (i.e., $y=x$). In this setting, they established an upper bound of $0.8822\cdot 3$ on the random access expectation. However, by allowing the multiplicity to scale proportionally with $x$, i.e., $y=\alpha x$, a strictly smaller expectation can be achieved. Our analysis, which  follow the same steps as in~\cite{GMZ24}, shows this improvement and is proven in Appendix~\ref{firstApen}.

\begin{restatable}{lemma}{betterAnalKEqThree}\label{theor_k=3}
It holds that $\lim_{x\to {}\infty} T_{\max}(G_3({x,\alpha x}))\leq 3-\frac{\alpha}{3+3\alpha}-\frac{2+10\alpha+5\alpha^2}{9(1+\alpha)^2}+\frac{(1+2\alpha)^2(1+2\alpha)}{9(1+\alpha)^2(2+\alpha)}+\frac{2\alpha^2(9+7\alpha)}{9(1+\alpha)^2(3+2\alpha)^2}.$ 
\end{restatable}

A numerical optimization shows that the minimum of this expression is attained at approximately $\alpha=0.834$. Substituting this value into  Theorem~\ref{theor_k=3} yields the following corollary, which slightly improves the result from~\cite{GMZ24}.

\begin{corollary}
    $T_{\max}(G_3({x,0.834x}))\leq 2.645=0.881\overline{66}\cdot 3.$
\end{corollary}
    
{In a recent paper~\cite{BLLLRS25},  the authors managed to evaluate the exact expression of the expectation and proved that \begin{align*}
\lim_{x\to {}\infty} T_{\max}(G_3({x,\alpha x}))=\frac{153+543\alpha+805\alpha^2+611\alpha^3+234\alpha^4+36\alpha^5}{3(1+\alpha)^2(2+\alpha)(3+2\alpha)^2}.   
\end{align*}}

Their numerical optimization showed that the optimal value of  $\alpha$ is $\alpha^*\approx 0.833968$ and for that value $\lim_{x\to {}\infty} T_{\max}(G_3({x,\alpha^* x}))\approx 0.881542\cdot 3$.

\subsubsection{Four Information Strands}
In this subsection, we obtain a closed form expression to the value of $T_{\max}(G_4({x,y}))$ using Lemma~\ref{Lem:CalExpeUsingAlphas}. To apply Lemma~\ref{Lem:CalExpeUsingAlphas}, we first compute the values of $\alpha_i^s(G_4({x,y}))$ for all $s \in [6x+4y-1]$ and $i \in [k]$, as established in the next lemma.

\begin{lemma}\label{lem:calAlphasK=4}
    For all $i \in [k]$ we have 
    \begin{itemize}
        \item[(i)] $\alpha_i^1(G_4({x,y})) = y$;
        \item[(ii)] $\alpha_i^2(G_4({x,y}))=3({x\choose 2}+x\cdot y)+{y\choose 2}+y(6x+3y)$;
        \item[(iii)] $\alpha_i^3(G_4({x,y}))=3\left({3x+2y\choose 3}-{x+2y\choose 3}-2x\binom{y}{2}\right)-3({x\choose 3}+{x\choose 2}y+x{y\choose 2})+3x(xy+\binom{x}{2})+y{6x+3y \choose 2}+{y \choose 2}(6x+3y)+{y\choose 3};$ 
        \item[(iv)] $\alpha_i^s(G_4({x,y}))={6x+4y \choose s}-{3x+3y \choose s}-3x{x+2y\choose s-1}-3x^2\binom{y}{s-2}{-6x^2\binom{y}{s-2}}$ \quad 
        for $4 \le s \le 3x+3y$;
        \item[(v)] $\alpha_i^s(G_4({x,y})) = {6x+4y \choose s}$ \quad for $3x+3y+1 \le s \le 6x+4y-1$.
    \end{itemize}
    
\end{lemma}

\begin{proof} Suppose we want to recover ${\bf e}_1$. Recall that $E_{i,j}$ denotes all the columns that lie in $\langle {\bf e}_i,{\bf e}_{j}\rangle\setminus\{{\bf e}_i, {\bf e}_{j}\}$.

(i) is easy to see.

In order to prove (ii), we count the number of $2$-sets that contain ${\bf e}_1$ in their span. If neither of the two columns is ${\bf e}_1$, then in order to recover ${\bf e}_1$, the $2$-set must span an edge of the form $E_{1,j}$, and there are $3({x\choose 2}+x\cdot y)$ such $2$-sets. Moreover, there are a total of ${y\choose 2}+y(6x+3y)$ $2$-sets that contain ${\bf e}_1$.

For (iii), we start by counting the $3$-sets that do not contain ${\bf e}_1$. Consider a $3$-set which lies in the span of $\langle{\bf e}_1,{\bf e}_j,{\bf e}_{j'}\rangle$ for some $j,j' \in \{2,3,4\}$. The only such 3-sets that do not recover ${\bf e}_1$ are: a) those where all columns lie in $E_{j,j'}$, there are ${x+2y\choose 3}$ such 3-sets, and b) those where one column lies in $E_{1,j}$ and there are two copies of ${\bf e}_{j'}$, there are $2x\binom{y}{2}$ such $3$-sets. Note that in this way, every such $3$-set which lies in $\langle {\bf e}_1,{\bf e}_j\rangle$ was counted twice. For example if $j=2$, then the 3-sets contained in $E_{1,2}\cup\{e_2\}$ were counted as part of the 3-sets in $\langle{\bf e}_1,{\bf e}_2,{\bf e}_3\rangle$ but also as part of the 3-sets in $\langle{\bf e}_1,{\bf e}_2,{\bf e}_4\rangle$. There are ${x\choose 3}+{x\choose 2}y+x{y\choose 2}$ such sets. 

The remaining $3$-sets that recover ${\bf e}_1$ but do not contain ${\bf e}_1$ are the sets where one column lies in $E_{i,i'}$ for $i,i'\neq 1$, and the other two columns lie in $\langle {\bf e}_1,{\bf e}_j\rangle \setminus \{{\bf e}_1\}$, for $j\neq i,i'$, there are $3x(xy+\binom{x}{2})$ such 3-sets. Finally, we add the number of $3$-sets which contain ${\bf e}_1$, which is $y{6x+3y \choose 2}+{y \choose 2}(6x+3y)+{y\choose 3}$.

(iv) We count the number of $s$-sets that do not recover ${\bf e}_1$. If the first coordinate in all the columns of the $s$-set is $0$, then it does not recover ${\bf e}_1$. There are ${3x+3y \choose s}$ such sets. The other cases are: a) when there is one column from $E_{1,i}$ for some $i \in \{2,3,4\}$, and the other $s-1$ columns must be contained in ${\bf e}_j\cup {\bf e}_{j'}\cup E_{j,j'}$ where $j,j'\notin\{1,i\}$, there are $3x{x+2y\choose s-1}$ such sets, b) when there is one column from $E_{1,i}$, another column from $E_{1,i'}$ ($i\neq i'$) and $s-2$ copies of $e_j$ ($j\neq i,i'$), there are $3x^2\binom{y}{s-2}$ such $s$-sets, and the last one c) then there is one column from $E_{1,i}$, another column from $E{i,i'}$ and $s-2$ copies of $e_{i''}$, there are $6x^2\binom{y}{s-2}$ such $s$-sets. 

From the analysis of (iv), it is easy to verify that any $s$-set for $s \ge 3x+3y+1$ will recover ${\bf e}_1$.

All of the previous computations did not depend on ${\bf e}_1$, and thus we conclude that they are the same for all ${\bf e}_i$, $i \in [k]$.
\end{proof}

Although we have closed-form expressions for $\alpha_i^s(G_4({x,y}))$ for all $s \in [6x+4y-1]$ and $i \in [k]$, and therefore an explicit expression for $T_{\max}(G_4({x,y}))$ (see Lemma~\ref{Lem:CalExpeUsingAlphas}), computing its exact value remains challenging. Nevertheless, using the computer algebra system \texttt{magma}, 
we obtain the plot shown in Figure~\ref{fig:exp}, along with the result presented in Proposition~\ref{prop:t4}.

\begin{figure}[ht!]
\centering
\begin{tikzpicture}[scale=1]
\begin{axis}[legend style={at={(1,1)}, legend style={cells={align=left}}, anchor = north east, /tikz/column 2/.style={
                column sep=5pt}},
		legend cell align={left},
		width=15cm,height=9cm,
    xlabel={Ratio $\alpha=y/x$},
    xmin=0, xmax=2,
    ymin=0.86, ymax=0.93,
    xtick={0,0.2,0.4,0.6,0.8,1,1.2,1.4,1.6,1.8,2},
    ytick={0.86, 0.88,0.9,0.92},
    ymajorgrids=true,
    grid style=dashed,
    every axis plot/.append style={},  yticklabel style={/pgf/number format/fixed}
]
\addplot+[color=orange,mark=o,mark size=1pt,smooth]
coordinates {
(0.200000000000000000000000000000,0.921360002005163295485876131037)
(0.400000000000000000000000000000,0.898456120565654034213871942066)
(0.600000000000000000000000000000,0.885761652356542414101733328202)
(0.800000000000000000000000000000,0.878982174390625147409337038259)
(1.00000000000000000000000000000,0.875769193467499905364511870152)
(1.20000000000000000000000000000,0.874767522627664406610813162891)
(1.40000000000000000000000000000,0.875161106032401919177682114436)
(1.60000000000000000000000000000,0.876441524882104504595686930056)
(1.80000000000000000000000000000,0.878283308898363977923341448433)
(2.00000000000000000000000000000,0.880473598323692115801433086660)
};
\addplot+[color=red,mark=o,mark size=1pt,smooth]
coordinates {
(0.100000000000000000000000000000,0.918725575085422038481421938665)
(0.200000000000000000000000000000,0.904396704901844299197157546918)
(0.300000000000000000000000000000,0.893829585416945973767069828494)
(0.400000000000000000000000000000,0.886048082670368274850708646012)
(0.500000000000000000000000000000,0.880358405182580134748080365019)
(0.600000000000000000000000000000,0.876258360020796298560848845630)
(0.700000000000000000000000000000,0.873378536768028211677631491439)
(0.800000000000000000000000000000,0.871443290280167262646805082751)
(0.900000000000000000000000000000,0.870244312453722472697714224737)
(1.00000000000000000000000000000,0.869622391467193983610478743627)
(1.10000000000000000000000000000,0.869454605830405958199271367471)
(1.20000000000000000000000000000,0.869645193693406851366372774141)
(1.30000000000000000000000000000,0.870118950036237795259093938984)
(1.40000000000000000000000000000,0.870816389775256043441856492337)
(1.50000000000000000000000000000,0.871690162224574717708564294651)
(1.60000000000000000000000000000,0.872702364021037738088345977920)
(1.70000000000000000000000000000,0.873822505021154289379828787101)
(1.80000000000000000000000000000,0.875025954121865129487559624784)
(1.90000000000000000000000000000,0.876292741517771242836888669567)
(2.00000000000000000000000000000,0.877606628262726138018719337189)
    };
\addplot+[color=green,mark=o,mark size=1pt,smooth]
coordinates {
(0.100000000000000000000000000000,0.902131371389833988429223160362)
(0.200000000000000000000000000000,0.890399125535903801822504324417)
(0.250000000000000000000000000000,0.885785773256411033658807227317)
(0.350000000000000000000000000000,0.878509520397594611150687043611)
(0.400000000000000000000000000000,0.875675123403308850363274615561)
(0.450000000000000000000000000000,0.873283505036298848574502688318)
(0.500000000000000000000000000000,0.871277613072940911690655325812)
(0.550000000000000000000000000000,0.869608558523473260036655529640)
(0.650000000000000000000000000000,0.867118539880646985318064636236)
(0.700000000000000000000000000000,0.866229878616559933239296207246)
(0.750000000000000000000000000000,0.865541033444808406041347562645)
(0.800000000000000000000000000000,0.865028248683241611938292634323)
(0.850000000000000000000000000000,0.864670780473500119650059916844)
(0.900000000000000000000000000000,0.864450468769161601890954413581)
(0.950000000000000000000000000000,0.864351377282925650299482486535)
(1.00000000000000000000000000000,0.864359489472819964437403104290)
(1.10000000000000000000000000000,0.864649350346984949847673415671)
(1.15000000000000000000000000000,0.864910532667203502689114076256)
(1.25000000000000000000000000000,0.865622470462900609093560648344)
(1.30000000000000000000000000000,0.866058867953316107870776051555)
(1.35000000000000000000000000000,0.866540612341108044247635597784)
(1.40000000000000000000000000000,0.867062335197983930451895110342)
(1.45000000000000000000000000000,0.867619242485321894347213909321)
(1.50000000000000000000000000000,0.868207048003642164131935630959)
(1.55000000000000000000000000000,0.868821915380789575213218203098)
(1.60000000000000000000000000000,0.869460407389495262293436982757)
(1.65000000000000000000000000000,0.870119441572470936060709840024)
(1.70000000000000000000000000000,0.870796251309137151476868812104)
(1.75000000000000000000000000000,0.871488351588207593285589307259)
(1.80000000000000000000000000000,0.872193508859249574073665733672)
(1.85000000000000000000000000000,0.872909714427743843022287794012)
(1.90000000000000000000000000000,0.873635160935101654955730168373)
(1.95000000000000000000000000000,0.874368221530030821014026176671)
(2.00000000000000000000000000000,0.875107431392592113523606965511)
};

\addplot+[color=blue,mark=o,mark size=1pt,smooth]
coordinates {
(0.100000000000000000000000000000,0.900549803781153575109533712097)
(0.150000000000000000000000000000,0.894351011747036922106391019301)
(0.200000000000000000000000000000,0.889059673105487691783104319161)
(0.250000000000000000000000000000,0.884546622961283645178467470709)
(0.300000000000000000000000000000,0.880703534491325127261717442853)
(0.350000000000000000000000000000,0.877439109537999080661733491730)
(0.400000000000000000000000000000,0.874676046242588169130548336841)
(0.450000000000000000000000000000,0.872348609491750249438548181919)
(0.500000000000000000000000000000,0.870400672361168139636402586947)
(0.550000000000000000000000000000,0.868784128055359468059936022822)
(0.600000000000000000000000000000,0.867457595156318993481459501681)
(0.650000000000000000000000000000,0.866385356486963858518798994530)
(0.700000000000000000000000000000,0.865536485119270894990206446605)
(0.750000000000000000000000000000,0.864884121124365630262099728125)
(0.800000000000000000000000000000,0.864404870377796483441318787371)
(0.850000000000000000000000000000,0.864078302684998011850460018840)
(0.900000000000000000000000000000,0.863886531111008755146526346038)
(0.950000000000000000000000000000,0.863813858004242800204178481285)
(1.00000000000000000000000000000,0.863846476034552955185051493896)
(1.05000000000000000000000000000,0.863972214799546454067214546836)
(1.10000000000000000000000000000,0.864180325324944137280701414309)
(1.15000000000000000000000000000,0.864461296197115613444466390438)
(1.20000000000000000000000000000,0.864806696197009031134906559445)
(1.25000000000000000000000000000,0.865209039214624670722667751746)
(1.30000000000000000000000000000,0.865661667958373255081606853769)
(1.35000000000000000000000000000,0.866158653570145639581235206582)
(1.40000000000000000000000000000,0.866694708742810776810023900940)
(1.45000000000000000000000000000,0.867265112334187613705147728592)
(1.50000000000000000000000000000,0.867865643797654130586678543425)
(1.55000000000000000000000000000,0.868492526018185470586602243627)
(1.60000000000000000000000000000,0.869142375364652321354869033857)
(1.65000000000000000000000000000,0.869812157953352919254769569368)
(1.70000000000000000000000000000,0.870499151270957286589198908103)
(1.75000000000000000000000000000,0.871200910432904776072903845528)
(1.80000000000000000000000000000,0.871915238460328707575032494259)
(1.85000000000000000000000000000,0.872640160048438681978051527153)
(1.90000000000000000000000000000,0.873373898374940521389317319255)
(1.95000000000000000000000000000,0.874114854560935466329098819044)
(2.00000000000000000000000000000,0.874861589450793125179286398403)

    };
\legend{$x=5$, $x=10$,$x=100$,$x=1000$
}
\end{axis}
\end{tikzpicture}
\caption{\label{fig:exp} Normalized (by $k=4$) random access expectation $T_{\max}(G_4({x,y}))$ for various $x$ and ratios $\alpha=y/x$.}
\end{figure}

\begin{proposition} \label{prop:t4}
    We have $T(4) \le 0.863814 \cdot 4$.
\end{proposition}
Thus far, we have studied the random access expectation of our construction for $k=3$ and $k=4$, in the next subsection, we extend the analysis to general $k$, using the complete graph setup.

\subsubsection{The Asymptotic Expectation for General $k$}
As observed in Lemma~\ref{lem:calAlphasK=4}, explicitly calculating the $\alpha$'s can be very challenging, especially when $k$ increases. Therefore, in this section, we derive a closed-form expression for the asymptotic random access expectation of Construction~\ref{ConsGeneralK} (i.e., when the number of columns goes to infinity), for every $k$.
To achieve this, we look at the problem from a graph viewpoint and apply Claim~\ref{Claim:EqToGraphs} to evaluate the random access expectation. 

Assume we collect either edges or vertices from the graph $K_k$. Each edge is collected with probability $p=\frac{x}{ky+\binom{k}{2}x}$, and each vertex with probability $P=\frac{y}{ky+\binom{k}{2}x}$. The random access expectation of the first strand is then:
\begin{align}
E[\tau_1] = 1 + (1-P) + \sum_{r=2}^{\infty} \text{Pr}\left( \tau_1 > r \right)\label{eq:ExpecTailSum}.
\end{align}
The key assumption that allows us to compute the asymptotic 
expectation is the following observation.

\begin{observation}
As we are interested in the asymptotic behavior of the random access expectation, we can assume that the same copy of an edge cannot be collected twice. This assumption is justified because, by taking $x$ to infinity, we can make the probability of collecting the same copy twice arbitrarily small.    
\end{observation}

Next, we decompose the last term in Equation~(\ref{eq:ExpecTailSum}) into disjoint cases. For that purpose, let $H_r$ denote the multi-set of components from $K_k$ collected at round $r$. Then, vertex $1$ (or information strand $1$) cannot be recovered at round $r\geq 2$ only if one of the following conditions holds: \\(i) $H_r$ does not contain any components that involve vertex $1$, 
or \\(ii) Vertex $1$ belongs to a tree containing $\ell + 1 \leq k$ vertices.

In the next lemma, we calculate the expectation of each of the cases. 
\begin{theorem}\label{Theorem:ExConstrcution}
Under the assumption that the same copy of an edge cannot be collected twice, we have:
\begin{align*}
    &\E\left( \text{ Case } i\ \right)= \frac{ \left( 1 - (k-1) p - P \right)^2 }{ (k-1)p + P },\\
   &\E\left( \text{ Case } ii\ \right) =(k-1)p\frac{v(2-v)}{(1-v)^2}+ \sum_{\ell=2}^{k-1} \binom{k-1}{\ell} (\ell + 1)^{\ell-1} \ell ! p^{\ell} \frac{1}{\left( 1 - u \right)^{\ell+1}},
\end{align*}
where $v= \binom{k-2}{2}p + (k-2)P $ and $u = \binom{k-\ell-1}{2}p + (k-\ell-1)P$.
\\In particular, it holds that:
\begin{align}
  T(k) \leq \min\limits_{\substack{p,P:\\kP+\binom{k}{2}p=1}} 1 + (1-P) + \E\left( \text{ Case } i\ \right)+\E\left( \text{ Case } ii\ \right).  \label{eq:boundOnT(k)}
  \end{align}
\end{theorem}

\begin{proof}
We start with the first event:
\begin{align*}
\E\left( \text{ Case } i\ \right) &= \sum_{r=2}^{\infty} \left( 1 - (k-1)p - P \right)^{r} \\
&= \frac{ \left( 1 - (k-1) p - P \right)^2 }{ (k-1)p + P }.
\end{align*}
For every $r\geq 2$, this represents the probability that we don't collect the first vertex and that no edges which are collected are adjacent to the first vertex. Since there are $k-1$ edges that are adjacent to vertex $1$ in the graph $K_k$, the result follows.

For the second term:
\begin{align}
\E\left( \text{ Case } ii\ \right) =  \sum_{r=2}^{\infty} \sum_{\ell=1}^{\min\{k-1,r\}} \binom{k-1}{\ell} \left( \ell + 1 \right)^{\ell-1} \ell!   \binom{r}{\ell} p^{\ell} \left( \binom{k-\ell-1}{2}p + (k-\ell-1)P \right)^{r-\ell}, \label{eq:complicatedCase}   
\end{align} 

this addresses the scenario where vertex $1$ has not been collected and where vertex $1$ belongs to a tree containing $\ell + 1 \leq k$ vertices, which means that the tree has at most $\ell \leq k-1$ edges. Since the number of draws is $r$, we must also have that $\ell \leq r$. In the expression, we partition the collected components in $H_r$ into two disjoint multi-sets:
\\(a) The first multi-set contains the vertex $1$ and, by assumption, forms a tree on $\ell+1$ vertices.\\ (b) The second multi-set contains components which are not connected to the tree in the first set.

We begin by analyzing the number of ways to select the first multiset. There are $\binom{k-1}{\ell}$ ways to choose $\ell$ additional vertices from the remaining $k-1$ (excluding vertex 1), to form a group of $\ell+1$ vertices containing vertex $1$. By Cayley's formula, there are $\left( \ell +1 \right)^{\ell-1}$ distinct trees that can be generated given $\ell+1$ vertices. Each of these trees has $\ell$ edges, which may be collected in $(\ell)!$ different orders. Since each edge is collected independently with probability $p$, this gives the contribution of the $p^{\ell}$ term. 

For the second multiset, consider the $k-\ell-1$ not included in the first tree. Among these, there are $\binom{k-\ell-1}{2}$ possible edges, and each vertex can also be collected independently. Therefore, the contribution from these components over the remaining $r-\ell$ rounds is: 
$\left( \binom{k-\ell-1}{2}p + (k-\ell-1)P \right)^{r-\ell}$. Note that when $\ell=k-1$, this term vanishes, as there are no remaining vertices.

Next, we get rid of the infinite sum in equation~(\ref{eq:complicatedCase}). For shorthand let \begin{align*}
&u\doteq\binom{k-\ell-1}{2}p + (k-\ell-1)P,\\
&v\doteq\left( \binom{k-2}{2}p + (k-2)P \right).
\end{align*}
Clearly, $|u|,|v| < 1$. Then, by reorganizing terms, we have
\begin{align*}
\E\left( \text{ Case } ii\ \right) &=\sum_{r=2}^{\infty} \sum_{\ell=1}^{\min\{k-1,r\}} \binom{k-1}{\ell} \left( \ell + 1 \right)^{\ell-1} \ell!   \binom{r}{\ell} p^{\ell} \left( \binom{k-\ell-1}{2}p + (k-\ell-1)P \right)^{r-\ell}\\&=\sum_{r=2}^{\infty}(k-1)rp\left( \binom{k-2}{2}p + (k-2)P \right)^{r-1}+\sum_{r=2}^{\infty} \sum_{\ell=2}^{\min\{k-1,r\}} \binom{k-1}{\ell}  \left( \ell + 1 \right)^{\ell-1} \ell!   \binom{r}{\ell} p^{\ell} u^{r-\ell}
\\&=(k-1)p\frac{v(2-v)}{(1-v)^2}+\sum_{r=2}^{\infty} \sum_{\ell=2}^{\min\{k-1,r\}} \binom{k-1}{\ell}  \left( \ell + 1 \right)^{\ell-1} \ell!   \binom{r}{\ell} p^{\ell} u^{r-\ell}
\\&=(k-1)p\frac{v(2-v)}{(1-v)^2} + \sum_{\ell=2}^{k-1} \binom{k-1}{\ell} (\ell + 1)^{\ell-1} \ell ! p^{\ell} \sum_{r=\ell}^{\infty} \binom{r}{\ell} u^{r-\ell}  \\
&= (k-1)p\frac{v(2-v)}{(1-v)^2}+ \sum_{\ell=2}^{k-1} \binom{k-1}{\ell} (\ell + 1)^{\ell-1} \ell ! p^{\ell} \frac{1}{\left( 1 - u \right)^{\ell+1}},
\end{align*}
as desired.
\end{proof}

Setting $k=4$ and $\alpha=0.95$ (the value of $\alpha=y/x$ that we used in Proposition~\ref{prop:t4}) in Theorem~\ref{Theorem:ExConstrcution} gives the upper bound $T(4)\leq 0.86375\cdot 4$. This bound is slightly stronger than the bound in Proposition~\ref{prop:t4}, as expected, since Theorem~\ref{Theorem:ExConstrcution} assumes that the same copy of an edge cannot be collected more than once. Nonetheless, as $x\to\infty$, the probability of resampling the same copy of an edge more than once tends to zero, and thus the two bounds converge.  

\begin{figure}[ht!]
\centering
\begin{tikzpicture}[scale=1]
\begin{axis}[legend style={at={(1,1)}, legend style={cells={align=left}}, anchor = north east, /tikz/column 2/.style={
                column sep=5pt}},
		legend cell align={left},
		width=15cm,height=9cm,
    xlabel={Ratio $\alpha=y/x$},
    xmin=0, xmax=100,
    ymin=0.790, ymax=0.822,
    xtick={1,20,40,60,80,100,120,140,160,180,200,220,240,260,280},
    ytick={0.79,0.8,0.81,0.82,0.85},
    ymajorgrids=true,
    grid style=dashed,
    every axis plot/.append style={},  yticklabel style={/pgf/number format/fixed}
]

\addplot+[color=red,mark=o,mark size=1pt,smooth]
coordinates {
(1.00000000000000000000000000000,0.817901988459934721348614557134)
(4.00000000000000000000000000000,0.804712653618753737526122690139)
(7.00000000000000000000000000000,0.799667470274803992521304187686)
(10.0000000000000000000000000000,0.798274378433089392522142495057)
(13.0000000000000000000000000000,0.798816673928879625784086580759)
(16.0000000000000000000000000000,0.800478766887243153849526204775)
(19.0000000000000000000000000000,0.802818021908012113142228770825)
(22.0000000000000000000000000000,0.805572739258578474350741598292)
(25.0000000000000000000000000000,0.808579023465157903524137947781)
(28.0000000000000000000000000000,0.811730062338531668060755441949)
(31.0000000000000000000000000000,0.814954293682416709283125219506)
(34.0000000000000000000000000000,0.818202872087259365506301216672)
(37.0000000000000000000000000000,0.821442083057203859257584745829)
(40.0000000000000000000000000000,0.824648554483925701674907827611)
(43.0000000000000000000000000000,0.827806129470522120796696941591)
(46.0000000000000000000000000000,0.830903766393368313501898091489)
(49.0000000000000000000000000000,0.833934095772217568447509870801)
(52.0000000000000000000000000000,0.836892409128536663029919571263)
(55.0000000000000000000000000000,0.839775938881433482779053648445)
(58.0000000000000000000000000000,0.842583338412623541056623417254)
(61.0000000000000000000000000000,0.845314302285374262367015296669)
(64.0000000000000000000000000000,0.847969286133877857139235487280)
(67.0000000000000000000000000000,0.850549298402737587165840378342)
(70.0000000000000000000000000000,0.853055744501838594895958979105)
(73.0000000000000000000000000000,0.855490309600246564594363013536)
(76.0000000000000000000000000000,0.857854870165689211138086280849)
(79.0000000000000000000000000000,0.860151427061324998057173177483)
(82.0000000000000000000000000000,0.862382054922023140342951945898)
(85.0000000000000000000000000000,0.864548863898467008457790878794)
(88.0000000000000000000000000000,0.866653970845220316433658338413)
(91.0000000000000000000000000000,0.868699477750539920667192383712)
(94.0000000000000000000000000000,0.870687455737849196524094061358)
(97.0000000000000000000000000000,0.872619933364532606829897914608)
(100.000000000000000000000000000,0.874498888240337120211825796064)

    };
\addplot+[color=green,mark=o,mark size=1pt,smooth]
coordinates {
(1.00000000000000000000000000000,0.820030000233562816818289321273)
(4.00000000000000000000000000000,0.814073495499325828950838048039)
(7.00000000000000000000000000000,0.809748679306911712523569955291)
(10.0000000000000000000000000000,0.806457208750971691455916370559)
(13.0000000000000000000000000000,0.803879773077447298096778327451)
(16.0000000000000000000000000000,0.801827159222403902073434290222)
(19.0000000000000000000000000000,0.800177803754114475731152212119)
(22.0000000000000000000000000000,0.798848590066561472472340712527)
(25.0000000000000000000000000000,0.797779878704764107365978652613)
(28.0000000000000000000000000000,0.796927222842404050512516067209)
(31.0000000000000000000000000000,0.796256481819644366169701984358)
(34.0000000000000000000000000000,0.795740781780987365400784089296)
(37.0000000000000000000000000000,0.795358539229904320360528538544)
(40.0000000000000000000000000000,0.795092126662272218383127717591)
(43.0000000000000000000000000000,0.794926942511744098620690347804)
(46.0000000000000000000000000000,0.794850744953540503083176263280)
(49.0000000000000000000000000000,0.794853163320661138086441541359)
(52.0000000000000000000000000000,0.794925332348672807059214435242)
(55.0000000000000000000000000000,0.795059613398958018575936549982)
(58.0000000000000000000000000000,0.795249378575343831189250079060)
(61.0000000000000000000000000000,0.795488841170722800665772234746)
(64.0000000000000000000000000000,0.795772920812938694763992656089)
(67.0000000000000000000000000000,0.796097134988649579873762197732)
(70.0000000000000000000000000000,0.796457510890378974273876946209)
(73.0000000000000000000000000000,0.796850513113339635165308109638)
(76.0000000000000000000000000000,0.797272983850769421853714870803)
(79.0000000000000000000000000000,0.797722093045174242365933702197)
(82.0000000000000000000000000000,0.798195296543888819565591573292)
(85.0000000000000000000000000000,0.798690300744946801803378403651)
(88.0000000000000000000000000000,0.799205032547118563287609277177)
(91.0000000000000000000000000000,0.799737613666358151002080317873)
(94.0000000000000000000000000000,0.800286338570999782265298843414)
(97.0000000000000000000000000000,0.800849655434921613824271141839)
(100.000000000000000000000000000,0.801426149622383376082558048430)

};

\addplot+[color=blue,mark=o,mark size=1pt,smooth]
coordinates {
(1.00000000000000000000000000000,0.820767073471689864891409054719)
(4.00000000000000000000000000000,0.817034847437058404294991368941)
(7.00000000000000000000000000000,0.814043506410305544931821488159)
(10.0000000000000000000000000000,0.811564249666994582710626737083)
(13.0000000000000000000000000000,0.809461362213521239769111024171)
(16.0000000000000000000000000000,0.807648373431878013351302693163)
(19.0000000000000000000000000000,0.806066830343318364375605222683)
(22.0000000000000000000000000000,0.804675256472163777646260512681)
(25.0000000000000000000000000000,0.803443031199398852201052098948)
(28.0000000000000000000000000000,0.802346801417376408912815198782)
(31.0000000000000000000000000000,0.801368272386644721013497617216)
(34.0000000000000000000000000000,0.800492788584874273730122206574)
(37.0000000000000000000000000000,0.799708387702502305670398506177)
(40.0000000000000000000000000000,0.799005149415420542969500311152)
(43.0000000000000000000000000000,0.798374734317919403734712496264)
(46.0000000000000000000000000000,0.797810049365545735372254369814)
(49.0000000000000000000000000000,0.797304999811893355800933102333)
(52.0000000000000000000000000000,0.796854301732962563746463671719)
(55.0000000000000000000000000000,0.796453337920232560266488505143)
(58.0000000000000000000000000000,0.796098045423957633959250283361)
(61.0000000000000000000000000000,0.795784826599736469507740991792)
(64.0000000000000000000000000000,0.795510477884233276374609531910)
(67.0000000000000000000000000000,0.795272132135541705506431634507)
(70.0000000000000000000000000000,0.795067211486485930098959150417)
(73.0000000000000000000000000000,0.794893388441946884554782531872)
(76.0000000000000000000000000000,0.794748553510771126457377262642)
(79.0000000000000000000000000000,0.794630788068556142806925267085)
(82.0000000000000000000000000000,0.794538341445834179128146075644)
(85.0000000000000000000000000000,0.794469611458120503780631822211)
(88.0000000000000000000000000000,0.794423127761372395528568647862)
(91.0000000000000000000000000000,0.794397537543531413768222439501)
(94.0000000000000000000000000000,0.794391593160511022983391948184)
(97.0000000000000000000000000000,0.794404141400756901528633503298)
(100.000000000000000000000000000,0.794434114121777461381660645183)

    };
\addplot+[color=orange,mark=o,mark size=1pt,smooth]
coordinates {
(1.00000000000000000000000000000,0.821794537940308947559410066925)
(4.00000000000000000000000000000,0.820698509575033928875271053659)
(7.00000000000000000000000000000,0.819708181923932604556231676011)
(10.0000000000000000000000000000,0.818802038839726066113041907679)
(13.0000000000000000000000000000,0.817964635093223155104377906062)
(16.0000000000000000000000000000,0.817184578527760323036952716431)
(19.0000000000000000000000000000,0.816453255565843344424121688482)
(22.0000000000000000000000000000,0.815764005458826133142123690307)
(25.0000000000000000000000000000,0.815111572142344284155545469964)
(28.0000000000000000000000000000,0.814491731912682581497342652616)
(31.0000000000000000000000000000,0.813901034991726823513450837380)
(34.0000000000000000000000000000,0.813336622461028625683397982103)
(37.0000000000000000000000000000,0.812796094097780778796194825451)
(40.0000000000000000000000000000,0.812277411255416421448779784048)
(43.0000000000000000000000000000,0.811778824312600750232203293280)
(46.0000000000000000000000000000,0.811298817642046624268842513146)
(49.0000000000000000000000000000,0.810836067274054797653867451075)
(52.0000000000000000000000000000,0.810389407897272066729162050590)
(55.0000000000000000000000000000,0.809957806823991669552063794588)
(58.0000000000000000000000000000,0.809540343218664776003668349163)
(61.0000000000000000000000000000,0.809136191352817707031436075231)
(64.0000000000000000000000000000,0.808744606975569856671481762447)
(67.0000000000000000000000000000,0.808364916120817411391412773098)
(70.0000000000000000000000000000,0.807996505839171526526648579469)
(73.0000000000000000000000000000,0.807638816464503500992832887632)
(76.0000000000000000000000000000,0.807291335114730188521871028461)
(79.0000000000000000000000000000,0.806953590193384556695286738865)
(82.0000000000000000000000000000,0.806625146708891946829974611240)
(85.0000000000000000000000000000,0.806305602266763783879704564015)
(88.0000000000000000000000000000,0.805994583619290955110298309026)
(91.0000000000000000000000000000,0.805691743680040947564238163674)
(94.0000000000000000000000000000,0.805396758928184284648811404150)
(97.0000000000000000000000000000,0.805109327141603828264748802073)
(100.000000000000000000000000000,0.804829165408767072180204186860)
(101.000000000000000000000000000,0.804737347393612865012165400887)
(104.000000000000000000000000000,0.804466469105414318964608813067)
(107.000000000000000000000000000,0.804202266384328446751448729639)
(110.000000000000000000000000000,0.803944509133822962881328489443)
(113.000000000000000000000000000,0.803692980299282695000211535595)
(116.000000000000000000000000000,0.803447474817335469402456913009)
(119.000000000000000000000000000,0.803207798674070607563622688756)
(122.000000000000000000000000000,0.802973768058484165338337870218)
(125.000000000000000000000000000,0.802745208599481901179101929729)
(128.000000000000000000000000000,0.802521954676437205841637265213)
(131.000000000000000000000000000,0.802303848794697665425324366980)
(134.000000000000000000000000000,0.802090741018609128955222782701)
(137.000000000000000000000000000,0.801882488455619104987350982566)
(140.000000000000000000000000000,0.801678954785863481902686339296)
(143.000000000000000000000000000,0.801480009832357460667046788160)
(146.000000000000000000000000000,0.801285529167523979598632481572)
(149.000000000000000000000000000,0.801095393752317789260559484532)
(152.000000000000000000000000000,0.800909489604654658404217880404)
(155.000000000000000000000000000,0.800727707494244468964437275590)
(158.000000000000000000000000000,0.800549942661263718813888300618)
(161.000000000000000000000000000,0.800376094556595114371365158398)
(164.000000000000000000000000000,0.800206066601616105789468161431)
(167.000000000000000000000000000,0.800039765965739922713037811856)
(170.000000000000000000000000000,0.799877103360106549884483354400)
(173.000000000000000000000000000,0.799717992845991049157242777805)
(176.000000000000000000000000000,0.799562351656645991862776153144)
(179.000000000000000000000000000,0.799410100031426313463337365761)
(182.000000000000000000000000000,0.799261161061161021576442315780)
(185.000000000000000000000000000,0.799115460543838908058673417793)
(188.000000000000000000000000000,0.798972926849766470421392546328)
(191.000000000000000000000000000,0.798833490795437123916969814844)
(194.000000000000000000000000000,0.798697085525422760353161944305)
(197.000000000000000000000000000,0.798563646401662883092597463504)
(200.000000000000000000000000000,0.798433110899583869489174393212)

};
\legend{
$k=100$,
$k=500$,$k=1000$,$k=5000$
}
\end{axis}
\end{tikzpicture}
\caption{\label{fig:infThe3} Normalized (by $k$) upper bound on the random access expectation $T(k)$ for different values for $k$ and ratios $\alpha=y/x$ from Theorem~\ref{Theorem:ExConstrcution}.}
\end{figure}

Now,  by substituting different values of $p$ and $P$ in Theorem~\ref{Theorem:ExConstrcution}, we can derive an upper bound on $T(k)$, and consequently on $\liminf_{k\to\infty}\frac{T(k)}{k}$. Computations in Magma suggest that the optimum of the expression in equation~\ref{eq:boundOnT(k)}, is achieved when $p\neq P$ (as can be seen in Figure~\ref{fig:infThe3}). Unfortunately, such asymmetric choices make the analysis cumbersome. To keep the calculations tractable, we adopt the symmetric choice $p=P=\frac{1}{k+\binom{k}{2}}=\frac{2}{k^2+k}$. With this choice the forthcoming corollary yields the best upper bounds currently known for both $T(k)$ and $\liminf_{k\to\infty}\frac{T(k)}{k}$.

\begin{corollary} \label{cor:ubfin}
It holds that:
    \begin{align*}
        T(k)&\leq 1+1-\frac{2}{k^2+k}+\frac{(k-1)^2}{2(k+1)}\\&+(k-1)\frac{2}{k^2+k}\cdot\frac{k^2-3k+2}{4k-2}\cdot\frac{k^2+5k-2}{4k-2}
       \\& +\sum_{\ell=2}^{k-1} \binom{k-1}{\ell}\frac{1}{(\ell + 1)^{2}}\ell!\frac{2^{\ell}k(k+1) }{(2k-\ell)^{\ell+1}}.
    \end{align*}
    Moreover: 
    \begin{align*}
        \liminf_{k\to\infty}\frac{T(k)}{k}\leq \frac{\pi^2}{12}\approx0.822467.
    \end{align*}
\end{corollary}

\begin{proof}
    For $p=P=\frac{1}{k+\binom{k}{2}}=\frac{2}{k^2+k}$ we have by Lemma~\ref{Theorem:ExConstrcution}: 
\begin{align*}
    T(k)\leq 1 + 1-\frac{2}{k^2+k}+ \E\left( \text{ Case } i\ \right)+\E\left( \text{ Case } ii\ \right). 
\end{align*}
    We now evaluate the expectation of both cases when  $p=P=\frac{2}{k^2+k}$. We start with the first case.
    \begin{align*}
    \E\left( \text{ Case } i\ \right)= \frac{(1-kp)^2}{kp}=\frac{(\frac{k-1}{k+1})^2}{\frac{2}{k+1}}=\frac{(k-1)^2}{2(k+1)}.
    \end{align*}
    Which goes to $k/2$ when $k\to\infty$. 
    
We now move on to case $ii$, we have:
\begin{align*}
  &v= \binom{k-2}{2}p +(k-2)P=(k-2)p(\frac{k-3}{2}+1)=\frac{(k-2)(k-1)}{k(k+1)}.
\end{align*}
Thus, the first term of $\E\left( \text{ Case } ii\ \right)$ when $p=P$ equals to \begin{align*}
    (k-1)p\frac{v(2-v)}{(1-v)^2}&=(k-1)\frac{2}{k^2+k}\cdot\frac{k^2-3k+2}{k^2+k}\cdot\frac{k^2+5k-2}{k^2+k}\cdot \frac{(k^2+k)^2}{(4k-2)^2}\\&=(k-1)\frac{2}{k^2+k}\cdot\frac{k^2-3k+2}{4k-2}\cdot\frac{k^2+5k-2}{4k-2},
\end{align*}
which goes to $\frac{1}{8}k$ when $k\to\infty$. In addition, similar to $v$, we have that \begin{align*}
    &u=\frac{(k-\ell-1)(k-\ell)}{k(k+1)},\\
    &1-u=\frac{2k+2\ell k-\ell^2-\ell}{k^2+k}.
\end{align*}
Hence, the second term of $\E\left( \text{ Case } ii\ \right)$ when $p=P$, equals to:
\begin{align*}
    &\sum_{\ell=2}^{k-1} \binom{k-1}{\ell} (\ell + 1)^{\ell-1} \ell ! p^{\ell} \frac{1}{\left( 1 - u \right)^{\ell+1}}\\=&\sum_{\ell=2}^{k-1} \binom{k-1}{\ell} (\ell + 1)^{\ell-1} \ell ! \frac{2^\ell}{k^\ell(k+1)^\ell} \frac{k^{\ell+1}(k+1)^{\ell+1}}{\left((\ell+1)(2k-\ell)\right)^{\ell+1}}\\=&\sum_{\ell=2}^{k-1} \binom{k-1}{\ell}\frac{1}{(\ell + 1)^{2}}\ell!\frac{2^{\ell}k(k+1) }{(2k-\ell)^{\ell+1}}.
\end{align*}

Together with what we did earlier, this gives us the required bound on $T(k)$. 
To obtain the bound on  $\liminf_{k\to\infty}\frac{T(k)}{k}$, we analyze the asymptotic behavior of the normalized above expression when $k\to\infty$. 
\begin{align*}
&\lim_{k\to\infty}\frac{1}{k}\sum_{\ell=2}^{k-1} \binom{k-1}{\ell}\frac{1}{(\ell + 1)^{2}}\ell!\frac{2^{\ell}k(k+1) }{(2k-\ell)^{\ell+1}}=\\
&\lim_{k\to\infty}\sum_{\ell=2}^{k-1} \frac{1}{k}\frac{(k+1)!2^\ell}{(k-1-\ell)!(\ell+1)^2(2k-\ell)^{\ell+1}}.
    \end{align*}

To evaluate the above expression, we aim to interchange the limit and the summation. To justify this step, we use the Dominated Convergence Theorem. 
Let $a_{k,\ell}$ denote the summand, defined as $$a_{k,\ell}=\frac{1}{k}\frac{(k+1)!2^\ell}{(k-1-\ell)!(\ell+1)^2(2k-\ell)^{\ell+1}}.$$ Our goal is to show that $$\lim_{k\to\infty}\sum_{\ell=2}^{k-1}a_{k,\ell}=\sum_{\ell=2}^{\infty}\lim_{k\to\infty} a_{k,\ell}.$$ To that end, we begin by analyzing the asymptotic behavior of $ a_{k,\ell}$ as $k\to\infty$. 

Using standard approximations, we observe that  
\[
\frac{(k+1)!}{(k-1-\ell)!} \sim k^{\ell+2}, \quad \text{and} \quad (2k - \ell)^{\ell+1} \sim (2k)^{\ell+1}.
\]  
Therefore,
\begin{align}
\lim_{k\to\infty} a_{k,\ell}=  \lim_{k\to\infty} \frac{2^\ell k^{\ell+2}}{(\ell+1)^2 k (2k)^{\ell+1}}=\frac{1}{2(\ell+1)^2}. \label{eq:a_kl}   
\end{align}

Next, we aim to find a summable sequence $M_\ell$ such that $|a_{k,\ell}| \leq M_\ell$ for all $k$.

From the approximation in\ref{eq:a_kl}, we have
\[
a_{k,\ell} \sim \frac{1}{2(\ell + 1)^2},
\]
which decays rapidly as $\ell \to \infty$. Moreover, for fixed $\ell$, the sequence $a_{k,\ell}$ is eventually decreasing in $k$.

Hence, there exists a constant $C$ (independent of $k$) such that
\[
|a_{k,\ell}| \leq \frac{C}{(\ell + 1)^2}.
\]

Since $\sum_{\ell=2}^\infty \frac{1}{(\ell + 1)^2}<\infty$, we conclude that $\{a_{k,\ell}\}$ is dominated by a summable sequence.

Thus, by the Dominated Convergence Theorem, we can interchange the limit and the summation:
\begin{align*}
\lim_{k \to \infty} \sum_{\ell = 2}^{k - 1} a_{k,\ell}
&= \sum_{\ell = 2}^\infty \lim_{k \to \infty} a_{k,\ell}
= \sum_{\ell = 2}^\infty \frac{1}{2(\ell + 1)^2}\\&=\frac{1}{2}\sum_{m=3}^\infty \frac{1}{m^2}=\frac{1}{2}(-1-\frac{1}{4}+\sum_{m=1}^\infty \frac{1}{m^2})\\&=
\frac{1}{2}(\frac{\pi^2}{6}-\frac{5}{4})=\frac{\pi^2}{12}-\frac{5}{8}.
\end{align*}

Together with all the previous expressions, we obtain 
$$\liminf_{k\to\infty}\frac{T(k)}{k}\leq \frac{1}{2}+\frac{1}{8}+\frac{\pi^2}{12}-\frac{5}{8}=\frac{\pi^2}{12}\approx0.822467.$$

\end{proof}
Figure~\ref{fig:infCor3} displays the normalized upper bound on $T(k)$ provided by Corollary~\ref{cor:ubfin}.

\begin{figure}[ht!]
\centering
\begin{tikzpicture}[scale=1]
\begin{axis}[legend style={at={(1,1)}, legend style={cells={align=left}}, anchor = north east, /tikz/column 2/.style={
                column sep=5pt}},
		legend cell align={left},
		width=15cm,height=9cm,
    xlabel={Dimension $k$},
    xmin=3, xmax=200,
    ymin=0.8, ymax=0.884,
    xtick={},
    ytick={},
    ymajorgrids=true,
    grid style=dashed,
    every axis plot/.append style={},  yticklabel style={/pgf/number format/fixed}
]

\addplot+[color=red,mark=o,mark size=1pt,smooth]
coordinates {
(3,0.882222222222222222222222222222)
(4,0.863789619551524313429075333838)
(5,0.852487505206164098292378175760)
(6,0.844943782635437034604796260410)
(7,0.839607002402903110117506978085)
(8,0.835667854475782189031910869236)
(9,0.832664775022565631978050866039)
(10,0.830316315755294557588102252495)
(11,0.828441672770154356868887771833)
(12,0.826919734312255878925596567128)
(13,0.825666538754737098293824080350)
(14,0.824622183207456883430472096091)
(15,0.823742874845906997281595016313)
(16,0.822995920132036811801556621509)
(17,0.822356462930810103078476982725)
(18,0.821805301204137274591237979346)
(19,0.821327389635444733718088760466)
(20,0.820910790391958431981137992236)
(21,0.820545923730375226487593367449)
(22,0.820225023530960684250414812743)
(23,0.819941735584299354125655724316)
(24,0.819690817043218266422681065026)
(25,0.819467908694195217182185057641)
(26,0.819269360394563308547748128718)
(27,0.819092095834013055100229824785)
(28,0.818933506731762069913474109207)
(29,0.818791369311186612751450737723)
(30,0.818663777806901910513265275828)
(31,0.818549091117671128558392226742)
(32,0.818445889694883176538141178811)
(33,0.818352940466141951608744619378)
(34,0.818269168115029534640641848454)
(35,0.818193631425102558473012061961)
(36,0.818125503686030527612121245217)
(37,0.818064056378772010572658853022)
(38,0.818008645523493510186117761338)
(39,0.817958700201982438695196741524)
(40,0.817913712865315104466570405280)
(41,0.817873231114626003786156233336)
(42,0.817836850703232754595545133223)
(43,0.817804209556002940298541096062)
(44,0.817774982639627810537695698138)
(45,0.817748877547598887080250388914)
(46,0.817725630687842943208663407746)
(47,0.817705003980440249054722103583)
(48,0.817686781988616705360666547126)
(49,0.817670769419026024625907690452)
(50,0.817656788937817840201177166965)
(51,0.817644679257586756351923624594)
(52,0.817634293457381919697517693230)
(53,0.817625497503816099927218379604)
(54,0.817618168946177429035159016792)
(55,0.817612195762499349151411083846)
(56,0.817607475336932034566116987567)
(57,0.817603913551599970492268120281)
(58,0.817601423978521169995134395980)
(59,0.817599927159181322282770260120)
(60,0.817599349961064213238681399423)
(61,0.817599625001889718354131808947)
(62,0.817600690133544924028787458724)
(63,0.817602487978747367381550371340)
(64,0.817604965514380797677479030136)
(65,0.817608073696217108350613918575)
(66,0.817611767120402984820317385092)
(67,0.817616003717662872184020521169)
(68,0.817620744476664886064118174916)
(69,0.817625953193424830475683526204)
(70,0.817631596243995261538379590514)
(71,0.817637642378009726632642370672)
(72,0.817644062530933823944760317874)
(73,0.817650829653120415519034743223)
(74,0.817657918553981150223496116313)
(75,0.817665305759774631656073664415)
(76,0.817672969383676697113155519181)
(77,0.817680889006943430033027239575)
(78,0.817689045570105340112573300967)
(79,0.817697421273243868784940421188)
(80,0.817705999484500952130823200495)
(81,0.817714764656060468541854230332)
(82,0.817723702246918455404588438566)
(83,0.817732798651828243181626484718)
(84,0.817742041135868205031317702730)
(85,0.817751417774134591092511693622)
(86,0.817760917396110721726140097688)
(87,0.817770529534307360954056628107)
(88,0.817780244376807996970608642379)
(89,0.817790052723387559607788118928)
(90,0.817799945944904276123030307738)
(91,0.817809915945692319291417524163)
(92,0.817819955128707997484033549927)
(93,0.817830056363204793169030333290)
(94,0.817840212954732853796270843732)
(95,0.817850418617276823724116640364)
(96,0.817860667447362395071539356858)
(97,0.817870953899976841111804587294)
(98,0.817881272766162247856145335867)
(99,0.817891619152152328161931779071)
(100,0.817901988459934721348614557127)
(101,0.817912376369130668254776888893)
(102,0.817922778820093012095296905165)
(103,0.817933191998131702962092731142)
(104,0.817943612318783461719605582660)
(105,0.817954036414049061721443095887)
(106,0.817964461119527880590314144141)
(107,0.817974883462385018579179908233)
(108,0.817985300650091427839124618330)
(109,0.817995710059882195789611164497)
(110,0.818006109228882418330973262138)
(111,0.818016495844854023096337810784)
(112,0.818026867737520493659529110870)
(113,0.818037222870429733480721266734)
(114,0.818047559333318321177071646932)
(115,0.818057875334943171488307226739)
(116,0.818068169196349151640818720798)
(117,0.818078439344543531068808920413)
(118,0.818088684306550282026012581139)
(119,0.818098902703819216148156738491)
(120,0.818109093246966752551813054305)
(121,0.818119254730826780204419461986)
(122,0.818129386029791613419953082399)
(123,0.818139486093424455622861015604)
(124,0.818149553942326093143054677287)
(125,0.818159588664239746989141676482)
(126,0.818169589410379124687883221983)
(127,0.818179555391965744009666754601)
(128,0.818189485876962552672889105241)
(129,0.818199380186991749266531288018)
(130,0.818209237694425526422613024606)
(131,0.818219057819639212975214553851)
(132,0.818228840028416992268430104383)
(133,0.818238583829501023314451853824)
(134,0.818248288772275394170445560961)
(135,0.818257954444576896372193300787)
(136,0.818267580470625128895959708295)
(137,0.818277166509064922997361736371)
(138,0.818286712251114528219747582845)
(139,0.818296217418813417463220105397)
(140,0.818305681763363957634266745845)
(141,0.818315105063561554235929341933)
(142,0.818324487124308215313245052796)
(143,0.818333827775204794280205280104)
(144,0.818343126869217464016844967602)
(145,0.818352384281414247797393376849)
(146,0.818361599907767687528322338744)
(147,0.818370773664019967761282455354)
(148,0.818379905484607036219561874620)
(149,0.818388995321638469262337011138)
(150,0.818398043143930024846276490855)
(151,0.818407048936086007087040720690)
(152,0.818416012697628736358849281497)
(153,0.818424934442172577816437688160)
(154,0.818433814196640130036636569275)
(155,0.818442652000518314856079739908)
(156,0.818451447905152240074694273060)
(157,0.818460201973074829101231452456)
(158,0.818468914277370326392630513313)
(159,0.818477584901069895198301912602)
(160,0.818486213936577625140864794462)
(161,0.818494801485125361989344938402)
(162,0.818503347656254861020363605969)
(163,0.818511852567325848999068096813)
(164,0.818520316343048658398984433679)
(165,0.818528739115040171348063024623)
(166,0.818537121021401880243230496216)
(167,0.818545462206318937302632546820)
(168,0.818553762819679126788523385878)
(169,0.818562023016710751481166597124)
(170,0.818570242957638479444928999382)
(171,0.818578422807356248416023542827)
(172,0.818586562735116373456574382619)
(173,0.818594662914234049047819483764)
(174,0.818602723521806479709793519914)
(175,0.818610744738445913697588233652)
(176,0.818618726748025892486326713624)
(177,0.818626669737440064759356575663)
(178,0.818634573896372947588607260214)
(179,0.818642439417082049565666726927)
(180,0.818650266494190800921974559529)
(181,0.818658055324491764274189693126)
(182,0.818665806106759626646905731279)
(183,0.818673519041573498953617138877)
(184,0.818681194331148073246321421659)
(185,0.818688832179173210856898817073)
(186,0.818696432790661556126735204696)
(187,0.818703996371803790827364244419)
(188,0.818711523129831163682074015810)
(189,0.818719013272884947670084602352)
(190,0.818726467009892495090732832385)
(191,0.818733884550449576741082608823)
(192,0.818741266104708707069049640104)
(193,0.818748611883273171854802878035)
(194,0.818755922097096488892187160948)
(195,0.818763196957387045332695260251)
(196,0.818770436675517667857965646625)
(197,0.818777641462939893701298043816)
(198,0.818784811531102721780364189060)
(199,0.818791947091375633866094445798)
(200,0.818799048354975685828574070871)

    };

\legend{
}
\end{axis}
\end{tikzpicture}
\caption{\label{fig:infCor3} Normalized (by $k$) upper bound on the random access expectation $T(k)$ for different values for~$k$ from Corollary~\ref{cor:ubfin}.}
\end{figure}

\section{Conclusion and Future Research}\label{sec:conclusion}
In this paper, we addressed the random access problem in DNA storage. We determined the exact values of $T_q(2)$ and $T(2)$ and we extended the construction presented in~\cite{GMZ24}. Our analysis demonstrated that this generalization achieves the best-known results to date. Despite these advancements, several intriguing questions remain open, which we aim to explore in future work. A few of these directions are the following:
\begin{enumerate}
 
    \item  From~\cite{BSGY24}, we know that for positive integers $a$ and $k$, it holds that $\frac{T(ak)}{ak}\leq \frac{T(k)}{k}$, and we wonder whether it holds that $\frac{T(k+1)}{k+1}\leq \frac{T(k)}{k}$ for all $k$? If that is true, what is the limit of $\frac{T(k)}{k}$ as $k$ tends to infinity? 
    
    \item In all existing constructions that achieve a small random-access expectation, including those in the present work, the number of columns of weight $m$ is the same for every set of $m$ coordinates. In other words, the marginal distribution of vectors of weight $m$ is uniform.  However, it had not been formally proven that this uniformity is the optimal structure for minimizing $T_{\max}(G)$. For $k=2$, we have shown that the optimal matrix adheres to this uniform structure, but is this true for general $k$? Even when we restrict ourselves to matrices whose columns have weight at most two (as we did in Construction~\ref{ConsGeneralK}), it is still unknown whether the symmetric pattern is truly best. We conjecture that it is, and establishing this rigorously is an important direction for future work.
    
    \item The best lower bounds on $T_q(n,k)$  to date were obtained in the seminal work~\cite{BSGY24}: $T_q(n,k)\geq n-\frac{n(n-k)}{n}(H_n-H_{n-k})$ and $T_q(n,k)\geq \frac{k+1}{2}$. Despite their significance, these bounds remain well below the best upper bounds currently known. Sharpening them—or closing the gap from the upper-bound side—constitutes an important avenue for future research.

     \item In this work, we analyzed only the substitution of $p = P$ in Theorem~\ref{Theorem:ExConstrcution}. A natural direction for future research is to optimize $p$ and $P$ given $k$, in order to minimize the expression in Theorem~\ref{Theorem:ExConstrcution} and obtain a better upper bound on $T(k)$.
    \end{enumerate}

\newpage

\appendix
\section{Appendix}\label{firstApen}
For the reader's convenience, we restate the statement of the lemma before presenting the corresponding proof.

\betterAnalKEqThree*
\begin{proof}
   By Corollary 4.8 in~\cite{GMZ24} we have:
   \begin{align*}
        T_{\max}(G_3({x,\alpha x}))&=3+\frac{2}{3x(1+\alpha)-2}-\frac{\alpha x-1}{3x(1+\alpha)-1}-\frac{2(\alpha x^2 +{x \choose 2})+\alpha x(3x+2\alpha x)+\frac{1}{2}\alpha x(\alpha x-1)}{{{3x(1+\alpha)-1}\choose{2}}}\\&+\sum_{s=3}^{x+2\alpha x}\prod_{i=0}^{s-1}\frac{x+2\alpha x-i}{3x(1+\alpha) -i-1}+\sum_{s=3}^{\alpha x+1} \frac{2x{\alpha x \choose {s-1}}}{{3x(1+\alpha)-1\choose s}}.  
   \end{align*} 
   We are going to evaluate each term separately. First:
   \begin{align*}
    \lim_{x\to{}\infty} \frac{2}{3x(1+\alpha)-2}=0.    
   \end{align*}
   Second:
   \begin{align*}
       \lim_{x\to {}\infty} \frac{\alpha x-1}{3x(1+\alpha)-1}=\frac{\alpha}{3+3\alpha}. 
   \end{align*}
   Third: \begin{align*}
       \frac{2(\alpha x^2 +{x \choose 2})+\alpha x(3x+2\alpha x)+\frac{1}{2}\alpha x(\alpha x-1)}{{{3x(1+\alpha)-1}\choose{2}}}
       &=\frac{2\alpha x^2+x^2-x+3\alpha x^2+2\alpha^2x^2+\frac{1}{2}\alpha^2x^2-\frac{1}{2}\alpha x}{\frac{1}{2}(3x(1+\alpha)-1)(3x(1+\alpha)-2)}\\&=\frac{(1+5\alpha+2.5\alpha^2)x^2-(1+0.5\alpha)x}{\frac{1}{2}(3x(1+\alpha)-1)(3x(1+\alpha)-2)},
   \end{align*}
which goes to $\frac{2+10\alpha+5\alpha^2}{9(1+\alpha)^2}$ when $x$ goes to infinity.

Next, we have\begin{align*}
    \sum_{s=3}^{x+2\alpha x}\prod_{i=0}^{s-1}\frac{x+2\alpha x-i}{3x(1+\alpha) -i-1}=\sum_{s=3}^{x+2\alpha x}\frac{x+2\alpha x}{3x(1+\alpha)-1}\cdot \frac{x+2\alpha x-1}{3x(1+\alpha)-2}\prod_{i=2}^{s-1}\frac{x+2\alpha x-i}{3x(1+\alpha) -i-1}
\end{align*}
Since for $i\geq 2$ we have that $\frac{x+2\alpha x-i}{3x(1+\alpha) -i-1}\leq \frac{1+2\alpha}{3+3\alpha}$, we obtain that \begin{align*}
     \sum_{s=3}^{x+2\alpha x}\prod_{i=0}^{s-1}\frac{x+2\alpha x-i}{3x(1+\alpha) -i-1}\leq  \frac{x+2\alpha x}{3x(1+\alpha)-1}\cdot \frac{x+2\alpha x-1}{3x(1+\alpha)-2}\sum_{s=3}^{x+2\alpha x}(\frac{1+2\alpha}{3+3\alpha})^{s-2}.
\end{align*}
Now we are going to evaluate $\sum_{s=3}^{x+2\alpha x}(\frac{1+2\alpha}{3+3\alpha})^{s-2}$. It holds that
\begin{align*}
   \sum_{s=3}^{x+2\alpha x}(\frac{1+2\alpha}{3+3\alpha})^{s-2}&=\sum_{s=1}^{x+2\alpha x-2}(\frac{1+2\alpha}{3+3\alpha})^{s}=\sum_{s=0}^{x+2\alpha x-2}(\frac{1+2\alpha}{3+3\alpha})^{s}-1\\&=\frac{(\frac{1+2\alpha}{3+3\alpha})^{x+2\alpha x-1}-1}{\frac{1+2\alpha}{3+3\alpha}-1}-1= \frac{(\frac{1+2\alpha}{3+3\alpha})^{x+2\alpha x-1}-1}{\frac{-2-\alpha}{3+3\alpha}}-1. 
\end{align*}
Combining all of the above, we obtain that \begin{align*}
 &\lim_{x\to {}\infty} \sum_{s=3}^{x+2\alpha x}\prod_{i=0}^{s-1}\frac{x+2\alpha x-i}{3x(1+\alpha) -i-1}\\&\leq \lim_{x\to {}\infty} \frac{x+2\alpha x}{3x(1+\alpha)-1}\cdot \frac{x+2\alpha x-1}{3x(1+\alpha)-2}(\frac{(\frac{1+2\alpha}{3+3\alpha})^{x+2\alpha x-1}-1}{\frac{-2-\alpha}{3+3\alpha}}-1)\\&=\frac{1+2\alpha}{3(1+\alpha)}\cdot \frac{1+2\alpha}{3(1+\alpha)}\cdot (\frac{3+3\alpha}{2+\alpha}-1)=\frac{(1+2\alpha)^2}{9(1+\alpha)^2}\cdot(\frac{1+2\alpha}{2+\alpha}).
\end{align*}
And we are left with the last term.

\begin{align*}
\sum_{s=3}^{\alpha x+1} \frac{2x{\alpha x \choose {s-1}}}{{3x(1+\alpha)-1\choose s}}&=^{(a)}2x\sum_{s=3}^{\alpha x+1}\frac{\frac{s}{\alpha x+1}{\alpha x+1\choose s}}{{3x(1+\alpha)-1\choose s}}=^{(b)}2x\sum_{s=3}^{\alpha x+1}\frac{s}{\alpha x+1}\frac{{3x(1+\alpha)-1-s\choose \alpha x +1-s}}{{3x(1+\alpha)-1\choose \alpha x+1}}\\&=\frac{2x}{\alpha x+1} {3x(1+\alpha)-1\choose \alpha x+1}^{-1}\sum_{s=3}^{\alpha x+1} s {3x(1+\alpha)-1-s\choose \alpha x +1-s},  
\end{align*}

Where $(a)$ holds because ${\alpha x \choose {s-1}}=\frac{s}{\alpha x+1}{\alpha x+1\choose s}$ and $(b)$ is due to $${\alpha x+1 \choose s}{3x(1+\alpha)-1\choose \alpha x+1}={3x(1+\alpha)-1-s\choose \alpha x+1-s}{3x(1+\alpha)-1\choose s}.$$

Now we are going to evaluate the term: $\sum_{s=3}^{\alpha x+1} s {3x(1+\alpha)-1-s\choose \alpha x +1-s}$, by placing $T=\alpha x+1-s$ we have that: \begin{align*}
    \sum_{s=3}^{\alpha x+1} s {3x(1+\alpha)-1-s\choose \alpha x +1-s}&= \sum_{T=0}^{\alpha x-2}(\alpha x +1-T){3x(1+\alpha)-1+T-\alpha x-1\choose T}\\&=\sum_{T=0}^{\alpha x-2}(\alpha x +1){3x+2\alpha x-2+T\choose T}-\sum_{T=0}^{\alpha x-2}T {3x+2\alpha x-2+T\choose T}.
\end{align*}
Now, by applying Pascal's triangle recursively on ${3x+3\alpha x-3\choose \alpha x-2}$ we have that
$$\sum_{T=0}^{\alpha x-2}{3x+2\alpha x-2+T\choose T}={3x+3\alpha x-3\choose \alpha x-2}$$ and hence we obtain:
\begin{align*}
 \sum_{T=0}^{\alpha x-2}(\alpha x +1){3x+2\alpha x-2+T\choose T}&=(\alpha x+1){3x+3\alpha x-3\choose \alpha x-2}\\&=(\alpha x+1){3x+3\alpha x-3\choose \alpha x-1}\frac{\alpha x -1}{3x+2\alpha x-1}.   
\end{align*}

In addition, we have:
\begin{align*}
    T {3x+2\alpha x-2+T\choose T}&=T\frac{(3x+2\alpha x-2+T)!}{(T!)(3x+2\alpha x-2)!}\\
    &=(3x+2\alpha x-1)\frac{(3x+2\alpha x-2+T)!}{(T-1)!(3x+2\alpha x-1)!}\\
    &=(3x+2\alpha x-1){3x+2\alpha x-2+T\choose T-1}.
\end{align*}
Hence we can rewrite:

\begin{align*}
    \sum_{T=0}^{\alpha x-2}T {3x+2\alpha x-2+T\choose T}&=(3x+2\alpha x-1)\sum_{T=1}^{\alpha x-2}{3x+2\alpha x-2+T\choose T-1}\\
    &=(3x+2\alpha x-1)\sum_{T=0}^{\alpha x-3}{3x+2\alpha x-1+T\choose T}\\
    &=(3x+2\alpha x-1){3x+3\alpha x-3\choose \alpha x-3}\\
    &=\frac{(\alpha x-2)(\alpha x-1)}{3x+2\alpha x}{3x+3\alpha x-3\choose \alpha x-1}.
\end{align*}

Combining all of the above we obtain
\begin{align*}
    \sum_{s=3}^{\alpha x+1} \frac{2x{\alpha x \choose {s-1}}}{{3x(1+\alpha)-1\choose s}}&=\frac{2x}{\alpha x+1}{3x(1+\alpha)-1\choose \alpha x+1}^{-1}\Bigg[ \frac{(\alpha x+1)(\alpha x-1)}{3x+2\alpha x-1}{3(x+\alpha x-1)\choose \alpha x-1}\\
    &\quad-\frac{(\alpha x-2)(\alpha x-1)}{3x+2\alpha x}{{3(x+\alpha x-1)\choose \alpha x-1}}\Bigg]\\
    &=\frac{2x(\alpha x-1)}{\alpha x+1}{3x(1+\alpha)-1\choose \alpha x+1}^{-1}{3(x+\alpha x-1)\choose \alpha x-1}\left[ \frac{\alpha x+1}{3x+2\alpha x-1}-\frac{\alpha x-2}{3x+2\alpha x}\right]\\&=\frac{2x(\alpha x-1)}{\alpha x+1}\frac{(\alpha x+1)\alpha x}{(3x+3\alpha x-1)(3x+3\alpha x-2)}\left[ \frac{9x+7\alpha x-2}{(3x+2\alpha x-1)(3x+2\alpha x)}\right]
\end{align*}
which goes to $\frac{2\alpha^2(9+7\alpha)}{9(1+\alpha^2)(3+2\alpha)^2}$ when $x$ goes to infinity.
\end{proof}

\newpage
\bibliographystyle{IEEEtran}

\end{document}